\documentclass[11pt]{article}
\usepackage[utf8]{inputenc}
\usepackage[T1]{fontenc}

\usepackage{typearea}
\usepackage{setspace}

\usepackage{fullpage}

\usepackage{comment} 
\usepackage{bbm}
\usepackage{framed} 
\usepackage{url} 
\usepackage{complexity}
\usepackage{booktabs}
\usepackage{amsmath,amssymb}

\usepackage{amsthm}
\usepackage{thmtools} 
\usepackage{thm-restate}
\usepackage{nicefrac}
\usepackage{calc}
\usepackage{enumerate}
\usepackage{enumitem}
\usepackage[usenames,dvipsnames]{xcolor}
\usepackage{algorithm}
\usepackage{graphicx}
\usepackage[font=footnotesize,labelfont=bf]{subcaption}
\usepackage[font=footnotesize,labelfont=bf]{caption}
\usepackage[nobreak=true]{mdframed}
\usepackage{appendix}
\usepackage[noend]{algpseudocode}
\usepackage[colorinlistoftodos]{todonotes}
\usepackage{xr}
\usepackage{array}
\usepackage{xspace}
\definecolor{ForestGreen}{rgb}{0.1333,0.5451,0.1333}
\definecolor{DarkRed}{rgb}{0.8,0,0}
\definecolor{Red}{rgb}{1,0,0}
\usepackage[linktocpage=true,
pagebackref=true,colorlinks,
linkcolor=DarkRed,citecolor=ForestGreen,
bookmarks,bookmarksopen,bookmarksnumbered]{hyperref}

\usepackage[capitalise]{cleveref}
\crefrangelabelformat{enumi}{#3#1#4--#5#2#6}
\usepackage{parskip}
\usepackage{chngcntr}
\usepackage{mathtools,stackengine}
\usepackage{multirow}

\stackMath
\newcommand{\stackGeq}[1]{%
	\setbox0=\hbox{${}\mathrel{\stackon[-1pt]{\geq}{\scriptstyle\text{#1\strut}}}{}$}
	\xdef\tmpwd{\dimexpr\the\wd0\relax}
	\kern.5\tmpwd\mathclap{\box0}&\kern.5\tmpwd
}

\usepackage{nameref}

\usepackage{tikz}
\usetikzlibrary{patterns}

\allowdisplaybreaks

\DeclareMathOperator*{\expectation}{\mathbb{E}}
\let\poly\relax
\DeclareMathOperator*{\poly}{poly}

\renewcommand\R{\mathbb{R}}

\newcommand\eps{\epsilon}

\newcommand\norm[1]{\left\| #1 \right\|}

\newcommand{\expect}{\expectation\expectarg}
\DeclarePairedDelimiterX{\expectarg}[1]{[}{]}{%
	\ifnum\currentgrouptype=16 \else\begingroup\fi
	\activatebar#1
	\ifnum\currentgrouptype=16 \else\endgroup\fi
}

\DeclarePairedDelimiterX{\nicesetarg}[1]{\{}{\}}{%
	\ifnum\currentgrouptype=16 \else\begingroup\fi
	\activatebar#1
	\ifnum\currentgrouptype=16 \else\endgroup\fi
}

\newcommand{\innermid}{\nonscript\;\delimsize\vert\nonscript\;}
\newcommand{\activatebar}{%
	\begingroup\lccode`\~=`\|
	\lowercase{\endgroup\let~}\innermid 
	\mathcode`|=\string"8000
}

\newcommand\prob[1]{\Pr\left[#1\right]}

\newcommand\opt{\textsc{Opt}\xspace}

\newcommand\optr{\textsc{Opt}_{\textsc{Rec}}\xspace}

\newcommand\alg{\textsc{Alg}\xspace}

\counterwithin{equation}{section}

\usepackage{eqparbox}

\newcommand\wKL[2]{\textsc{KL}_w\left(#1 \mid \mid #2\right)}

\theoremstyle{plain}

\newtheorem{theorem}{Theorem}[section]

\newtheorem{lemma}[theorem]{Lemma}
\newtheorem{fact}[theorem]{Fact}

\newtheorem{claim}[theorem]{Claim}

\newtheorem{corollary}[theorem]{Corollary}

\newlength{\continueindent}
\usepackage{etoolbox}
\makeatletter
\newcommand*{\ALG@customparshape}{\parshape 2 \leftmargin \linewidth \dimexpr\ALG@tlm+\continueindent\relax \dimexpr\linewidth+\leftmargin-\ALG@tlm-\continueindent\relax}
\apptocmd{\ALG@beginblock}{\ALG@customparshape}{}{\errmessage{failed to patch}}
\makeatother

\makeatletter
\def\thm@space@setup{%
	\thm@preskip=\parskip \thm@postskip=0pt
}
\makeatother

\usepackage{etoolbox}
\usepackage{tikz}
\usetikzlibrary{tikzmark}
\usetikzlibrary{calc}

\errorcontextlines\maxdimen

\newcommand{\ALGtikzmarkcolor}{black}
\newcommand{\ALGtikzmarkextraindent}{4pt}
\newcommand{\ALGtikzmarkverticaloffsetstart}{-.5ex}
\newcommand{\ALGtikzmarkverticaloffsetend}{-.5ex}
\makeatletter
\newcounter{ALG@tikzmark@tempcnta}

\newcommand\ALG@tikzmark@start{%
	\global\let\ALG@tikzmark@last\ALG@tikzmark@starttext%
	\expandafter\edef\csname ALG@tikzmark@\theALG@nested\endcsname{\theALG@tikzmark@tempcnta}%
	\tikzmark{ALG@tikzmark@start@\csname ALG@tikzmark@\theALG@nested\endcsname}%
	\addtocounter{ALG@tikzmark@tempcnta}{1}%
}

\def\ALG@tikzmark@starttext{start}
\newcommand\ALG@tikzmark@end{%
	\ifx\ALG@tikzmark@last\ALG@tikzmark@starttext
	\else
	\tikzmark{ALG@tikzmark@end@\csname ALG@tikzmark@\theALG@nested\endcsname}%
	\tikz[overlay,remember picture] \draw[\ALGtikzmarkcolor] let \p{S}=($(pic cs:ALG@tikzmark@start@\csname ALG@tikzmark@\theALG@nested\endcsname)+(\ALGtikzmarkextraindent,\ALGtikzmarkverticaloffsetstart)$), \p{E}=($(pic cs:ALG@tikzmark@end@\csname ALG@tikzmark@\theALG@nested\endcsname)+(\ALGtikzmarkextraindent,\ALGtikzmarkverticaloffsetend)$) in (\x{S},\y{S})--(\x{S},\y{E});%
	\fi
	\gdef\ALG@tikzmark@last{end}%
}

\apptocmd{\ALG@beginblock}{\ALG@tikzmark@start}{}{\errmessage{failed to patch}}
\pretocmd{\ALG@endblock}{\ALG@tikzmark@end}{}{\errmessage{failed to patch}}
\makeatother

\algblock[with]{With}{EndWith}
\algblockdefx[With]{With}{EndWith}%
[1]{\textbf{with} #1 \textbf{do}}%
{}

\makeatletter
\ifthenelse{\equal{\ALG@noend}{t}}%
{\algtext*{EndWith}}
{}%
\makeatother

\global\long\def\inner#1#2{\left\langle #1,#2\right\rangle }

\title{Chasing Positive Bodies}

\author{}
 \author{Sayan Bhattacharya\thanks{Department of Computer Science, University of Warwick. Email:
 		\texttt{S.Bhattacharya@warwick.ac.uk}. Supported by Engineering and Physical Sciences Research Council, UK (EPSRC) Grant EP/S03353X/1.}
 	\and
 Niv Buchbinder\thanks{Department of Statistics and Operations Research, School of Mathematical Sciences, Tel Aviv University, Tel Aviv. Emails:
 		\texttt{niv.buchbinder@gmail.com}, \texttt{roiel@tauex.tau.ac.il}. Supported in part by Israel Science Foundation grant 2233/19 and United States - Israel Binational Science Foundation grant 2022418.}
   \and
   Roie Levin$^{\dag}$\thanks{Supported in part by a Fulbright Israel Postdoctoral Fellowship.} 
   \and
 	Thatchaphol Saranurak\thanks{Computer Science and Engineering Division, University of Michigan. Email: 		\texttt{thsa@umich.edu}. Supported by NSF CAREER grant 2238138.}}

\date{}

\begin{document}
\maketitle

\pagenumbering{gobble}

\begin{abstract}
We study the problem of chasing positive bodies in $\ell_1$: given a sequence of bodies $K_{t}=\{x^{t}\in\R_{+}^{n}\mid C^{t}x^{t}\ge1,P^{t}x^{t}\le1\}$ revealed online, where $C^{t}$ and $P^{t}$ are nonnegative matrices, the goal is to (approximately) maintain a point $x_t \in K_t$ such that $\sum_t \|x_t - x_{t-1}\|_1$ is minimized. This captures the fully-dynamic low-recourse variant of any problem that can be expressed as a mixed packing-covering linear program and thus also the fractional version of many central problems in dynamic algorithms such as set cover, load balancing, hyperedge orientation, minimum spanning tree, and matching.

We give an $O(\log d)$-competitive algorithm for this problem, where $d$ is the maximum row sparsity of any matrix $C^t$. This bypasses and improves exponentially over the lower bound of $\sqrt{n}$ known for general convex bodies. Our algorithm is based on iterated information projections, and, in contrast to general convex body chasing algorithms, is entirely memoryless.

We also show how to round our solution dynamically to obtain the first fully dynamic algorithms with \emph{competitive recourse} for all the stated problems above; i.e. their recourse is less than the recourse of every other algorithm on every update sequence, up to polylogarithmic factors. This is a significantly stronger notion than the notion of \emph{absolute recourse} in the dynamic algorithms literature. 
\end{abstract}

\newpage

\pagenumbering{arabic}

\section{Introduction}
\label{sec:intro}

We study the problem of \emph{chasing positive bodies} in $\ell_{1}$ defined as follows. We are given a sequence of bodies $K_{t}=\{x^{t}\in\R_{+}^{n}\mid C^{t}x^{t}\ge1,P^{t}x^{t}\le1\}$ revealed online, where $C^{t}$ and $P^{t}$ are matrices with non-negative entries. The goal is to (approximately) maintain a point $x^{t}\in K_{t}$ such that the total $\ell_{1}$-movement, $\sum_{t}\|x^{t}-x^{t-1}\|_{1}$, is minimized where $x^0=0$. 
More generally, given weight $w\in\R_{+}^{n}$, we want to minimize the weighted $\ell_{1}$-movement, $\sum_{t}w_{i}\sum_{i=1}^{n}|x_{i}^{t}-x_{i}^{t-1}|$. This captures the fully dynamic variant of any problem that can be expressed as a mixed packing-covering linear program, and thus also the fractional version of many central fully dynamic problems when the goal is to minimize \emph{recourse}, i.e., the amount of change to the solution, as measured by $\ell_1$ movement. Examples include set cover, load balancing (hyperedge orientation), minimum spanning tree, and matching.

A more general version our problem, called the \emph{convex body chasing} problem, allows $K_{t}$ to be an arbitrary convex body and the goal is to minimize total $\ell_{p}$-movement for any $p\ge1$. Friedman and Linial \cite{FL93} introduced convex body chasing as a vast generalization of many online problems. This problem has been the subject of intensive study recently \cite{BBEKU20,ABCGL19,BKLLS20} which has culminated in algorithms with $O(n)$ competitive ratios \cite{AGTG21,Sellke20}. This bound is nearly tight since there is a lower bound of $\Omega(\max\{\sqrt{n},n^{1-1/p}\})$ \cite{BKLLS20}. Unfortunately, because of this strong lower bound, algorithms for this general problem have generated no interesting applications to concrete combinatorial optimization problem.

In this paper, we show that for the special yet still expressive case of chasing positive bodies in $\ell_{1}$, we can bypass and exponentially improve upon the $\sqrt{n}$ lower bound. We then show that a solution to the positive body chasing problem can be rounded online to yield low-recourse fully dynamic algorithms for all the combinatorial problems mentioned above. Our algorithms have \emph{competitive recourse} guarantees, a significantly stronger notion of recourse than the usual one used in the dynamic algorithm literature, which we will discuss soon.

\subsection{Our Results}

Our main contribution is the following theorem.
\begin{theorem}
\label{thm:main}For any $\eps\in(0,1]$, there is an an $O\left(\nicefrac{1}{\eps}\log\left(\nicefrac{d}{\eps}\right)\right)$-competitive algorithm for chasing positive bodies in $\ell_{1}$ such that $x^{t}\in K_{t}^{1+\eps}=\left\{ x^{t}\in R_{+}^{n}\ |\ C^{t}x^{t}\geq1,P^{t}x^{t}\leq1+\eps\right\} $ at time $t$, and $d$ is the maximal number of non-negative coefficients in a covering constraint. 
\end{theorem}

Note we give our algorithm \emph{$\eps$-resource augmentation}, i.e. we allow it to violate the packing constraints slightly. Alternatively, by scaling the solution, we may produce a solution that fully satisfies all packing constraints but violates the covering constraints up to an $\eps$ factor. As we discuss below, $\eps$-resource augmentation does not matter in our applications since we lose additional approximation factors while rounding anyway. 
We compare the total movement of our algorithm that maintains a point in $K_{t}^{1+\eps}$ with an optimal solution that maintains a point in $K_{t}$. In several applications where $d=O(1)$, such as the fractional version of dynamic edge orientation and set cover with bounded frequency, our competitive ratio is completely independent of $n$. We also remark that we can handle static box constraints of the form $x\leq b$, where $b\in R_{+}^{n}\cup \{\infty\}$ without any violation. See \cref{sec:box} for more details.

We complement \cref{thm:main} by the following lower bound:
\begin{restatable}{theorem}{lbmainthm}
\label{thm:lb_mainthm}
    No algorithm for positive body chasing can achieve competitive ratio better than \[\Omega\left(\max\left\{\min\left(\frac{1}{\eps \sqrt{\log (1/\eps)}}, \sqrt{n}\right),\log n\right\}\right).\]
\end{restatable}

The $\Omega(\log n)$ lower bound follows by the known bound for covering LPs \cite{BN05, GKL21}. The $\Omega(\min(1/(\eps \sqrt{\log (1/\eps)}), \sqrt{n}))$ lower bound follows by a reduction to the $\sqrt{n}$ lower bound for the general convex bodies chasing problem (Theorem 5.4 of in \cite{BKLLS20})\footnote{We would like to thank Mark Sellke for pointing out this reduction.}. The full proof appears in \cref{sec:lb}.

This lower bound implies that all the assumptions of \Cref{thm:main} are crucial for the exponential improvement over the general convex body chasing case.
In particular, it shows that the lower bound of $\Omega(\sqrt{n})$ for the general convex body chasing problem holds for the positive body chasing problem (without resource augmentation) as well. More generally, our linear dependence on $1/\eps$ is essentially tight. For example, when $\eps = \Theta(n^{-\delta})$ for a constant $0<\delta<\nicefrac{1}{2}$, the competitive ratio is at least $\tilde{\Omega}(n^{\delta})$ (hiding logarithmic terms), and in particular $O(\poly \log n)$-competitive ratio is impossible. When $\eps= O(\frac{\log n}{n})$ the best algorithm is still the $O(n)$-competitive algorithm for general convex body chasing \cite{AGTG21,Sellke20}. Finally, we remark that for $\ell_p$-norms, when $p>1$, the known lower bound of \cite{FL93} forbids any $o(n^{1-1/p})$-competitive algorithms even for chasing positive bodies and even with $\Omega(1)$-resource augmentation.

\paragraph{Competitive Recourse for Dynamic Problems.}
An overarching goal of online and dynamic algorithms is to maintain near-optimal solutions to combinatorial problems as constraints change over time while minimizing the number of edits to the solution, a.k.a. recourse. An extensive line of work has produced low-recourse algorithms for Steiner tree under terminal updates \cite{DBLP:journals/siamdm/ImaseW91,DBLP:journals/siamcomp/GuG016,DBLP:conf/soda/GuptaK14,lkacki2015power,DBLP:conf/focs/GuptaL20}, load balancing under job arrivals \cite{DBLP:journals/jcss/AwerbuchAPW01,GKS14,KLS23}, set cover under element insertions/deletions \cite{GKKP17,DBLP:conf/stoc/AbboudA0PS19,  DBLP:conf/focs/BhattacharyaHN19,DBLP:conf/focs/GuptaL20, DBLP:conf/soda/BhattacharyaHNW21,DBLP:conf/esa/AssadiS21}, facility location under client updates \cite{BLP22,guo2020facility}, and many fully dynamic graph problems under edge updates like edge orientation \cite{brodal1999dynamic,sawlani2020near,bera2022new}, graph coloring \cite{solomon2020improved}, maximal independent sets \cite{assadi2018fully}, and spanners \cite{baswana2012fully,bhattacharya2022simple}.

The works mentioned above, and most others on dynamic algorithms, measure recourse in absolute terms. Such bounds are of the form ``after $T$ updates, the algorithm incurs recourse at most $kT$''. We call these \emph{absolute recourse} bounds. On the other hand, the online algorithms literature prefers \emph{competitive} analysis, where one compares the performance of the algorithm to the best algorithm in hindsight. In this paper we give \emph{competitive recourse} bounds: we say an algorithm has $c$-competitive recourse if it incurs recourse at most $c$ times that of any other \emph{offline} algorithm on every update sequence.

There are several advantages to studying competitive recourse over absolute recourse. An absolute recourse bound is a worst-case bound over \emph{all} update sequences, whereas a competitive recourse bound is tailored to \emph{each} update sequence. To illustrate this, consider an update sequence for set cover that repeatedly inserts and deletes the same element without changing the optimal solution. While competitive recourse algorithms would incur no recourse, an algorithm with absolute recourse guarantees might have large recourse proportional to these ``irrelevant'' updates. Hence, competitive recourse bounds can be much stronger than absolute bounds. For a concrete example, it is straightforward to obtain a fully dynamic algorithm for $(1+\eps)$-approximate matching with $O(1/\eps)$ absolute recourse by eliminating short augmenting paths, but no non-trivial algorithm with competitive recourse was known prior to our work. 

For some problems small absolute recourse may not be possible at all. In the context of fully dynamic load balancing, \cite{KLS23} show that no algorithm, even with full knowledge of the job arrival/departure sequence, can achieve competitive ratio $\alpha=o(\log(n))$ with less than $\Omega(n^{1/\alpha})$ absolute recourse. They write: 
\begin{quote}
``To circumvent the negative result, one needs to consider a different measurement for recourse for the fully dynamic model.'' 
\end{quote}
We argue in this paper that competitive recourse is the natural way around this obstacle.

\paragraph{Applications of \Cref{thm:main} to Combinatorial Problems. }

In this paper, we give the first fully dynamic algorithms with polylogarithmic competitive recourse for central problems studied in the dynamic algorithms literature, including set cover, load balancing, bipartite matching and minimum spanning tree. The results are summarized in \Cref{tab:app}. In the table, our algorithms with $\alpha$-approximation and $c$-competitive recourse guarantees have the following formal guarantee. For \emph{every} update sequence, our algorithms maintain $\alpha$-approximate solution with recourse at most $c$ times the optimal recourse for maintaining an optimal solution. More generally, each of our algorithms even allow the following trade-off. Given parameter $\beta$, for \emph{every} update sequence $\sigma$, it maintains a $(\alpha\beta)$-approximate solution with total recourse at most $c\cdot\optr^{\beta}(\sigma)$ where $\optr^{\beta}(\sigma)$ denotes the optimal recourse for maintaining a $\beta$-approximate solution undergoing updates $\sigma$. This is useful because it is natural to expect that the optimal recourse becomes significantly smaller when we allow some approximation.

\begin{table}[h]
\footnotesize{

\def\arraystretch{1.3}

\begin{tabular}{|>{\raggedright}m{0.097\textwidth}|>{\centering}m{0.12\textwidth}|>{\centering}m{0.15\textwidth}|>{\centering}m{0.14\textwidth}|>{\centering}m{0.075\textwidth}|>{\raggedright}m{0.08\textwidth}|>{\raggedright}m{0.145\textwidth}|}
\hline 
\textbf{Problems} & \textbf{Approx.} & \textbf{Competitive recourse} & \textbf{Randomized vs. Deterministic} & \textbf{Ref.} & \textbf{Update type} & \textbf{Variables}\tabularnewline
\hline 
\multirow{2}{0.13\textwidth}{Set cover} & $O(\log n)$ & $O(\log n\log f)$  & Randomized & Cor. \ref{cor:set cover} & \multirow{2}{0.08\textwidth}{elements} & \multirow{2}{0.19\textwidth}{$n$ elements,\\ each in $\le f$ sets}\tabularnewline
\cline{2-5} 
 & $O(f)$ & $O(f\log f)$ & Deterministic & Cor. \ref{cor:set cover} &  & \tabularnewline
\hline 
\multirow{2}{0.13\textwidth}{Load balancing} & $2+\eps$ & $O\left(\hspace{-0.05in}\begin{array}{cc}
     \frac{1}{\eps^{5}}\log n \\
     \cdot \log\frac{r}{\eps}\log\frac{1}{\eps}
\end{array}\hspace{-0.05in}\right)$ & Deterministic & Cor. \ref{cor:kt:makespan} & \multirow{2}{0.08\textwidth}{jobs} & \multirow{2}{0.19\textwidth}{$n$ jobs, each applicable\\ to $\le r$ machines}\tabularnewline
\cline{2-5} 
 & $O\left(\frac{\log\log n}{\log\log\log n}\right)$ & $O(\log r)$  & Randomized & Cor. \ref{cor:kt:makespan} &  & \tabularnewline
\hline 
Bipartite matching & $1+\eps$ & $O\left(\frac{1}{\eps^{4}}\log n\log\frac{n}{\eps}\right)$  & Randomized & Cor. \ref{cor:matching:rounding} & edges & $n$ vertices\tabularnewline
\hline 
Minimum spanning tree & $2+\eps$ & $O\left(\frac{1}{\eps^3}\log n\log\frac{n}{\eps}\right)$  & Randomized & Cor. \ref{cor:kt:mst} & edges & $n$ vertices\tabularnewline
\hline 
\end{tabular}

}

{\footnotesize{}\caption{\label{tab:app}Summary of our fully dynamic algorithms with competitive recourse. All randomized algorithms assume an oblivious adversary, and incur a polynomially-small additive term in the recourse bound which we omit. See \Cref{sec:rounding} for details.}
}{\footnotesize \par}
\end{table}

An important feature of our fractional algorithms is that their guarantees are independent of the type of dynamic update because \Cref{thm:main} works so long as the feasible region at every step forms a positive body. 
For example, our fractional minimum spanning tree and bipartite matching algorithms work under both vertex and edge arrivals/departures, our fractional load balancing algorithm works under fine-grained updates to job-machine loads, and our fractional set cover algorithm works when set costs can be updated. Since our rounding schemes also work under these generalized types of updates, our final algorithms do as well.

We emphasize that our focus is minimizing recourse and not \emph{update time}, the usual metric in the dynamic algorithms literature. Obtaining competitive recourse and small update time simultaneously is an interesting research direction.

\subsection{Connections to Previous Work}

\paragraph{Positive Body Chasing.}

Our result for chasing positive bodies directly generalizes the influential line of work on online covering problems~\cite{BN05,GN14,DBLP:conf/soda/GuptaL20}. In these, the convex bodies are nested and defined solely by covering constraints: for each time $t$, the body is $K_{t}=\left\{ x\in R_{+}^{n}\ |\ C^{t}x\geq1\right\} $ where $C^{t}$ is a non-negative matrix and $K_{t}\subseteq K_{t-1}$. The task is to maintain a monotonically increasing $x^{t}\in K_{t}$ while minimizing $\inner wx$ for fixed $w\in\R_{+}^{n}$. This goal is equivalent to minimizing the $\ell_{1}$-movement weighted by $w$, since decreasing any variable never helps to satisfy covering constraints.  Even this special case of positive body chasing has been amazingly successful in unifying previous results and in resolving important open questions in competitive analysis. This includes, e.g., the classic ski rental problem,  online set cover \cite{AAABN03}, weighted paging and variants \cite{BBN07,BBN08, CLNT22}, graph optimization problems related to connectivity and cuts \cite{AAABN04,NPS11}, the dynamic TCP-acknowledgement problem \cite{KKR01,BJN07}, metrical task systems~\cite{BNN10b,BCLL18a}, and the $k$-server problem~\cite{BCLL18,BGMN18}.

Another line of work \cite{ABFP13,ABC16}  generalizes \cite{BN05,GN14} in a different way. Here the convex bodies $K_{t}$ are nested covering constraints and $x^{t}\in K_{t}$ may only move monotonically. The goal in \cite{ABFP13} is to minimize the maximum violation of a \emph{fixed} set of packing constraints, and the goal of \cite{ABC16} is to minimize a non-decreasing convex function. This is different from our objective since we must handle an online sequence of both packing and covering constraints, but our solution $x$ is allowed to move non-monotonically.

\cite{bhattacharya2023dynamic} gave dynamic algorithms for maintaining $(1+\eps)$-approximate solutions to mixed packing-covering linear programs. They showed how to handle \emph{relaxing} updates (where the feasible region only grows) using small \emph{update time}, and asked as an open problem if it is possible to do the same for \emph{restricting updates} (where the feasible region only shrinks). \Cref{thm:main} resolves the recourse version of this question: we can maintain an $(1+\eps)$-approximate solution of the mixed packing-covering linear program with total absolute recourse $O(|x^{*}|_{1}\cdot\frac{1}{\eps}\log\frac{n}{\eps})$, where $x^{*}$ is any feasible solution at the end of the sequence of restricting updates. 
This follows because an offline algorithm can just move the solution to $x^{*}$ from the beginning.

\paragraph{Competitive Recourse in Combinatorial Problems. }
Here, we only compare our algorithms from \Cref{tab:app} with some of the few known dynamic algorithms with competitive recourse guarantees.

For the set cover problem, the online algorithms in \cite{GN14} support element insertions and guarantee $\min\{O(\log n\log f),O(f\log f)\}$-approximation, and these are the best possible. When every set has  unit weight, the approximation ratio also translates to a competitive recourse bound, because no set ever leaves the solution. Our algorithm (\Cref{cor:set cover}) is the first fully dynamic competitive-recourse algorithm, even in the unweighted case.

For the load balancing problem, our algorithms have interesting consequences even when we focus on special cases; \Cref{cor:kt:makespan} implies the first competitive-recourse algorithm for fully dynamic edge orientation in weighted graphs.\footnote{This corresponds to the load balancing instance where each job is applicable to two machines.} Previously, this was known only for unweighted graphs \cite{brodal1999dynamic}. In the unweighted setting, our algorithm with $O(\log r)$-competitive recourse is the first algorithm for fully dynamic hyperedge orientation with \emph{constant} competitive recourse and non-trivial approximation when the rank $r=O(1)$. The previous result of \cite{bera2022new} suffers a $O(\log n)$ factor in recourse.

\cite{GTW14} gave an algorithm for maintaining a spanning tree (more generally a matroid base) when edges have acquisition costs and evolving holding costs. Their result can be used to obtain an $O(\log^{2}n)$ competitive recourse algorithm for maintaining a spanning tree under edge updates but does not give guarantees about the cost of the tree at every step. Our algorithm (\Cref{cor:kt:mst}), in contrast, guarantees a guarantees $(2+\eps)$ approximation with $O_\eps(\log^2 n)$ competitive recourse.

The only other dynamic algorithm with competitive recourse guarantees we are aware of are \cite{azar2023competitive} for vertex coloring and \cite{avin2016online,avin2020dynamic} for balanced graph partitioning. We hope that our generic framework will facilitate further results in this direction.

\subsection{Techniques and Overview}

Our key technical contribution is an algorithm for chasing positive \emph{halfspaces}, or positive bodies defined by a single covering or packing constraint (we explain the reduction from chasing bodies to chasing halfspaces in \Cref{sec:reduction_halfspaces}). Our algorithm is the following. Given point $x^{t-1}$ from the last round, our algorithm performs a projection in KL divergence onto the new feasible region (plus a small additive update for technical reasons). This is sometimes referred to as an \emph{information projection} (e.g.~\cite{csiszar1984sanov}). One can interpret many online covering algorithms as information projections \cite{BGMN18}; our contribution is to extend the analysis for packing constraints which has until now remained elusive. Interestingly, in contrast to competitive algorithms for general convex body chasing (\cite{AGTG21,Sellke20}) this algorithm is completely memoryless, in that it only depends on $x^{t-1}$ and the current violated constraint.

To analyze the algorithm, we write an auxiliary linear program whose optimal value is the optimal recourse of an offline  algorithm with knowledge of all future constraints, and then fit a dual to our online algorithm's solution. We construct this dual using the Lagrange multipliers of the convex programs we use to compute our projections. We start with a warmup analysis in \cref{sec:the_algo,sec:the_analysis,sec:warmup}. The proof here is simple and contains the main ideas, but loses an additional aspect ratio term. To obtain the full result of \cref{thm:main}, we need to overcome a significant barrier: one can show that no monotone dual solution can avoid a dependence on the aspect ratio (see e.g. \cite{BN09}). Nevertheless, we show that it is possible to carefully \emph{decrease} some coordinates of the dual to achieve the desired bound. The details are in \cref{sec:removing_delta}.

We turn to rounding our competitive fractional solution for common applications in \cref{sec:rounding}. Our chasing positive bodies framework captures the task of maintaining fractional solutions to set cover (\cref{sec:setcov}), bipartite matching (\cref{sec:matching}), load balancing (\cref{sec:load_balance}), and minimum spanning tree (\cref{sec:mst}) under dynamic updates. We show that we can also round these fractional solutions to true combinatorial solutions with bounded loss in both recourse and approximation. Our rounding algorithms in \cref{sec:matching,sec:mst} reuse a common recipe: we first maintain an object we call a \emph{stabilizer}, which contains a good solution that itself has good competitive recourse (see \cref{app:matching} and \cref{app:th:kt:mst} for details). We then run existing \emph{absolute} recourse algorithms for maintaining a true solution within this stabilizer. Since the number of updates to the stabilizer is competitive with the recourse of the optimal offline solution, and the absolute recourse algorithms are competitive with respect to the number of updates to the stabilizer, overall this strategy yields a solution with good competitive recourse.

\section{Preliminaries}\label{sec:pre}

\textbf{Mathematical Notation.} All logarithms in this paper are taken to be base $e$. For $w \in \R$, we use the notation $w_+ = \max(0, w)$. We use the symbol $\oplus$ to denote symmetric difference of sets. In the following definitions, let $x,y \in \R^n_+$ be vectors. 
    The standard dot product between $x$ and $y$ is denoted
$\langle x, y\rangle = \sum_{i=1}^n x_i y_i$. We use a weighted generalization of KL divergence. Given
        a weight function $w$, define
	\[\wKL{x}{y} : = \sum_{i=1}^n w_i \left[x_i \log \left(\frac{x_i}{y_i}\right) -x_i + y_i\right].\]
It is known that $\wKL{x}{y} \geq 0$ for nonnegative vectors $x,y$ (one can check this is true term by term above).
Throughout this paper, we use the convention that for any $a,b \geq 0$, the expression 
\begin{align}
\left(a + \infty\right) \log \left(\frac{a + \infty}{b + \infty} \right) -a + b = 0. \label{eq:kl_limit}
\end{align}
These definitions agree with the respective limit behaviors.

\textbf{Simplifying Positive Body Chasing Instances.} 
In the positive body chasing problem, since $x^0=0$ and $x^t\geq 0$, then we have

\[\sum_{t \in [T]} \sum_{i \in [n]}  w_i (x_i^t - x^{t-1}_i)_+ \leq \sum_{t \in [T]} \sum_{i \in [n]} w_i \cdot |x^t_i -x^{t-1}_i| \leq 2\cdot \sum_{t \in [T]} \sum_{i \in [n]}  w_i (x_i^t - x^{t-1}_i)_+.\]

Therefore, without loss of generality, we use the \emph{upward recourse} as the measure of the movement of both our algorithm and the optimal algorithm.

Henceforth, we assume without loss of generality that at each time $t$, the positive body $K_t$ is defined by a single covering constraint $\inner{c^{t}}x\ge1$ or packing constraint $\inner{p^{t}}x\le 1$. This reduction from chasing bodies to chasing halfspaces is standard (see e.g. \cite{BBEKU20}), but we spell out the main argument in \Cref{sec:reduction_halfspaces} for completeness. 

Let $c_i^{\max} = \max_{t}c_i^t$, $c_i^{\min} = \min_{t | c_i^t\neq 0}c_i^t$. We define the aspect ratio $\Delta = \max_i c_i^{\max} / c_i^{\min}$. Let $d^t=|\{i \mid c_i^t\neq 0\}|$ be the sparsity of the covering constraint. Let $d=\max_t\{d^t\}$ be the maximal row sparsity of any covering constraint.

\section{An Algorithm for Chasing Positive Bodies}

\label{sec:positive_bodies}

In this section we design a fractional algorithm for chasing positive bodies proving \cref{thm:main}. As discussed in \cref{sec:pre} we may assume without loss of generality that at each time $t$ the positive body $K_t$ is defined by a single covering or packing constraint.

\subsection{The Algorithm:}

\label{sec:the_algo}

The algorithm is given a parameter $\eps\in(0,1]$. Whenever it gets a violated constraint, it projects back to the constraint using the following procedure.
\begin{itemize}
    \item Initially, $x_i^{0}=0$ for all $i$. At any time $t=1, 2, \ldots,T$:
    \item When a violated covering constraint $\inner{c^{t}}x<1$ arrives, set $x^t$ to be the solution to 
\begin{align*}
	\begin{array}{ccc}
			  & {\displaystyle \min_{x}\sum_{i| \, c_i^t\neq 0}w_{i}\cdot\left[ 
				\widehat x_{i}\log\left(\frac{\widehat x_{i}}{\widehat x_{i}^{t-1}}\right)-\widehat x_{i}\right] }\vspace{0.1cm} \\ 
			\text{s.t.} & \inner{c^{t}}x\ge1 & (I).\\
	\end{array}
\end{align*}
 where $\widehat x_i := x_i + \frac{\eps}{4d^t\cdot c_i^t}$, $\widehat x^{t-1}_i := x^{t-1}_i + \frac{\eps}{4d^t\cdot c_i^t}$, and $d^t= |\{i \mid c_i^t\neq 0\}|$ is the sparsity of the covering constraint.\footnote{Note: at time $t$ we redefine $\widehat x_i^{t-1}$, even though it was already defined at $t-1$ as $x_i^{t-1} + \eps / (3d^{t-1}\cdot c_i^{t-1})$. This is a slight abuse of notation but it makes our proofs cleaner, and also makes apparent the connection to KL divergence.} If for some $x_i^t$ the coefficients $c_i^t=0$, then we set $x_i^t=x_i^{t-1}$.
    \item When we are given a violated packing constraint  $\inner{p^{t}}x>1+\eps$, set $x^t$ to be the solution to 
    \begin{align*}
	\begin{array}{ccc}
			& {\displaystyle \min_{x}\sum_{i| \, p_i^t\neq 0}w_i\left(x_i^t\log \left(\frac{x_i}{x_i^{t-1}}\right)-x_i\right)} \vspace{0.1cm}\\
			\text{s.t.} & {\displaystyle \inner{p^{t}}x\leq1+\eps} & (II).\\
	\end{array}
\end{align*}
 If for some $x_i^t$ the coefficients $p_i^t=0$, then we set $x_i^t=x_i^{t-1}$.
\end{itemize}

We assume without loss of generality that we only get violated constraints.
Let $C, P\subseteq\{1, \ldots, T\}$ be the set of times in which covering/packing constraints appear respectively. 

\subsection{Analysis Framework}

\label{sec:the_analysis}

Let $y^{t}$, $z^{t}$ be the Lagrange multipliers of the constraints $(I)$ and $(II)$ above. In our proof we use the KKT conditions at the optimal solution of this program for every time-step $t$: 
\begin{align}
	&\forall t \in C & \inner{c^{t}}{x^{t}} & =1\label{eq:kkt1}\\
	&\forall t \in C, \ \forall i, c_i^t\neq 0 & w_{i}\log\left(\frac{\widehat x_{i}^{t}}{\widehat x_{i}^{t-1}}\right) & = c_{i}^{t}y^{t} \label{eq:kkt2}\\
	& \forall t \in P & \inner{p^{t}}{x^{t}} & =1+\eps\label{eq:kkt3}\\
	&\forall t \in P, \ \forall i & w_{i}\log\left(\frac{x_{i}^{t}}{x_{i}^{t-1}}\right) & = -p_{i}^{t}z^{t} \label{eq:kkt4}
\end{align}

For convenience, we define $y^0 = z^0 = 0$. 

The optimal upward recourse can be computed by following linear program $(\overline{\mathcal{P}})$.
\begin{align*}
	\boxed{\begin{array}{cc}
			& {\displaystyle (\overline{\mathcal{P}}): \min\sum_{t=1}^{T}\sum_{i=1}^{n}w_{i}\overline{\ell}_{i}^{t}}\\
			\text{s.t. }\\
			\\
			\forall t \in C: & {\displaystyle \inner{c^{t}}{\overline{x}^t}\geq1}\\
			\\
			\forall t \in P: & {\displaystyle \inner{p^{t}}{\overline{x}^t}\leq1}\\
			\\
			\forall i,t: & {\displaystyle \overline{x}_i^t-\overline{x}_{i}^{t-1}\leq\overline{\ell}_{i}^{t}}\\
			\\
			 & \overline{x}_i^t, \overline{\ell}_{i}^{t}\geq0
	\end{array}}\qquad\boxed{\begin{array}{cc}
	& {\displaystyle (\overline{\mathcal{D}}): \max\sum_{t \in C}\overline{y}^{t}-\sum_{t \in P} \overline z^{t}}\\
	\text{s.t. }\\
	\\
	\forall i,t\in C: & c_{i}^{t}{\displaystyle \overline{y}^{t}-\overline{r}_{i}^{t}+\overline{r}_{i}^{t+1}\leq 0}\\
	\\
	\forall i,t\in P: & {\displaystyle -p_{i}^{t} z^{t}-\overline{r}_{i}^{t}+\overline{r}_{i}^{t+1}\leq 0}\\
	\\
	& {\displaystyle 0\leq \overline{r}_{i}^{t}\leq w_{i}} \\ \\
	\forall t: & \displaystyle \overline{y}^t, \overline z^t \geq 0
\end{array}}
\end{align*}
 Our goal from now on is to construct a feasible solution to the dual program $(\overline{\mathcal{D}})$ that bounds the total cost of our online algorithm.

We start with a weaker bound than the one that appears in \cref{thm:main} as a warmup.

\subsection{Warmup: an $O\left(\nicefrac{1}{\eps} \cdot\log \left(\nicefrac{d \Delta}{\eps}\right)\right)$ competitive bound}
\label{sec:warmup}
In this section, we show our first bound on the recourse of our algorithm.
\begin{theorem}[Warmup Bound]
\label{thm:warmup}
The total recourse of the algorithm is at most $O\left(\frac{1}{\eps}\cdot \log\left(\frac{d\cdot \Delta}{\eps}\right)\right)$ times the objective of $(\overline{\mathcal{D}})$.
\end{theorem}
Recall that $c_i^{\max} = \max_{t}c_i^t$ and $c_i^{\min} = \min_{t | c_i^t\neq 0}c_i^t$, as well as the aspect ratio 
$\Delta = \max_i c_i^{\max} / c_i^{\min}$.

We construct the following dual solution.

Let $A= \log\left(1+\frac{4 d \cdot \Delta}{\eps}\right)$. We set
\begin{align}
    \bar{y}^{\tau} &= y^{\tau}/A, \nonumber \\
    \bar{z}^{\tau} & = z^{\tau}/A, \nonumber \\
    \overline{r}_{i}^{\tau} &=
    w_i\cdot\left(1- \frac{1}{A}\log\left(1+ \frac{4d c_i^{\max}\cdot x_i^{\tau-1}}{\eps}\right)\right) \nonumber
    \intertext{(which is useful to write as)}
    &= w_i \cdot \left(1-\frac{1}{A}\left(\log\left(x_i^{\tau-1}+ \frac{\eps}{4 d \cdot c_i^{\max}}\right)- \log\left(\frac{\eps}{4 d \cdot c_i^{\max}}\right)\right)\right).
\end{align}

We start by showing that our dual is feasible.
\begin{lemma}
    The variables $(\bar{y}, \bar{z}, \bar{r})$ are a feasible dual solution to  $(\overline{\mathcal{D}})$.
\end{lemma}

\begin{proof}
For $\tau\in C$ we have
\begin{align*}
\overline{r}_{i}^{\tau}-\overline{r}_{i}^{\tau+1} & = \frac{w_i}{A} \cdot \log\left(\frac{x_i^{\tau}+ \frac{\eps}{4 d \cdot c_i^{\max}}}{x_i^{\tau-1}+ \frac{\eps}{4d \cdot c_i^{\max}}}\right)\geq \frac{w_i}{A} \log\left(\frac{x_{i}^{\tau}+ \frac{\eps}{4 d^\tau \cdot c^\tau_i}}{x_{i}^{\tau-1}+ \frac{\eps}{4 d^\tau \cdot c^\tau_i}}\right) = c_i^{
\tau} \bar{y}^{\tau}.
\end{align*}
The inequality follows since for any $i$ and $\tau\in C$, we have $x_i^{\tau} \geq x_i^{\tau-1}$ and therefore, for any $a\geq b$ the inequality $\frac{x_i^{\tau}+ a}{x_i^{\tau-1}+ a}\leq \frac{x_i^{\tau}+ b}{x_i^{\tau-1}+ b}$ holds. Similarly, for $\tau\in P$ we have
\begin{align*}
\overline{r}_{i}^{\tau}-\overline{r}_{i}^{\tau+1} & = \frac{w_i}{A}\cdot \log\left(\frac{x_i^{\tau}+ \frac{\eps}{4 d \cdot c_i^{\max}}}{x_i^{\tau-1}+ \frac{\eps}{4d \cdot c_i^{\max}}}\right)\geq \frac{w_i}{A} \log\left(\frac{x_{i}^{\tau}}{x_{i}^{\tau-1}}\right) = -p_i^{
\tau} \bar{z}^{\tau}.
\end{align*}
The inequality follows since for any $\tau\in P$, we have $x_i^{\tau} \leq x_i^{\tau-1}$ and therefore, for any $a\geq 0$ 
the inequality $\frac{x_i^{\tau}}{x_i^{\tau-1}}\leq \frac{x_i^{\tau}+ a}{x_i^{\tau-1}+ a}$ holds. Finally,
\begin{align*}
    0 & \leq \overline{r}_{i}^{\tau} = w_i\cdot\left(1- \frac{1}{A}\log\left(1+ \frac{4d c_i^{\max}\cdot x_i^{\tau-1}}{\eps}\right)\right) \leq w_i.
\end{align*}
To see this, note that for all $\tau \in C$ and all $i\in[n]$ we have $0\leq x^{\tau}_i \leq \frac{1}{c^{\tau}_i}$, 
otherwise the constraint in round $\tau$ is already satisfied upon arrival. Hence
\[0\leq \log\left(1+ \frac{4d c_i^{\max}\cdot x_i^{\tau-1}}{\eps}\right) \leq \log\left(1+ \frac{4d c_i^{\max}}{\eps c^{\tau-1}_i}\right) \leq \log\left(1+\frac{4 d \cdot \Delta}{\eps}\right)=A. \qedhere\]
\end{proof}

Next, to relate the movement of the algorithm to the dual objective, we need the following pair of crucial lemmas. First we bound the recourse in terms of just the $y$ terms.

\begin{lemma}[Bounding movement in terms of $y$]
\label{lemmamovement}
    The upward movement cost of the algorithm at any time $t\in C$ can be bounded as
 \begin{align*}
	\sum_{i}w_{i}\left(x_{i}^{t}-x_{i}^{t-1}\right)_{+} & \leq \left(1+ \frac{\eps}{4}\right)y^{t}. 
 \end{align*}
\end{lemma}

\begin{proof}
From $1+x \leq e^x$ with $x = \log(b/a)$, we get the following convenient inequality (sometimes referred to as the ``Poor man's Pinkser'' inequality). For all $a,b \geq 0$,
\[a - b \leq a \log \left(\frac{a}{b}\right).\]

Using this fact, we deduce
\begin{align*}
	\sum_{i}w_{i}\left(x_{i}^{t}-x_{i}^{t-1}\right)_{+} &  = \sum_{i| \, c_i^t\neq 0}w_{i}\left(\widehat x_{i}^{t}-\widehat x_{i}^{t-1}\right)_{+}\\
 &\leq\sum_{i| \, x_i^t>x_i^{t-1}}w_{i}\cdot \widehat x_{i}^{t} \cdot\log\left(\frac{\widehat x_{i}^{t}}{\widehat x_{i}^{t-1}}\right) \\
	& \stackrel{\eqref{eq:kkt2}}{=} \sum_{i| \, x_i^t>x_i^{t-1}}\left(x_{i}^{t}+\frac{\eps}{4d^t\cdot c^t_i}\right)c_{i}^{t}y^{t}   \\
	& \leq \frac{\eps}{4} y^t + y^t \cdot \sum_{i }c_{i}^{t} x_{i}^{t}    \\
 &\stackrel{\eqref{eq:kkt1}}{=}\left(1+ \frac{\eps}{4}\right)y^{t}. \qedhere
\end{align*}

\end{proof}

To bound the recourse as a function of the full dual objective, we also need argue that subtracting the $z$ terms does not lose too much.

\begin{lemma}[Bounding $z$ in terms of $y$]
\label{lemmaKL}
The dual vectors $y$ and $z$ satisfy
\begin{align}
    \left(1+ \frac{\eps}{4}\right)\sum_{t \in C} y^t - (1+\eps)\sum_{t \in P} z^t & \geq 0
\end{align}
\end{lemma}

\begin{proof}
By non-negativity of KL we have
\begin{align*}
	&\forall t\in C & 0 & \leq \sum_{i| \, c_i^t\neq 0}w_{i}\left[\widehat x_{i}^{t}\log\left(\frac{\widehat x_{i}^{t}}{\widehat x_{i}^{t-1}}\right) -\widehat x_i^t + \widehat x_i^{t-1}\right] = \sum_{i| \, c_i^t\neq 0}w_{i}\left[\widehat x_{i}^{t}\log\left(\frac{\widehat x_{i}^{t}}{\widehat x_{i}^{t-1}}\right) - x_i^t +  x_i^{t-1}\right]\\
	&\forall t \in P & 0 & \leq \sum_{i| \, p_i^t\neq 0}w_{i}\left[x_i^t\log\left(\frac{x_{i}^{t}}{x_{i}^{t-1}}\right) -x_i^t + x_i^{t-1}\right].
\end{align*}
Summing these over all times $t$ we get:
\begin{align}
	0 & \leq \sum_{t \in C} \sum_{i| \, c_i^t\neq 0} w_i\left[\widehat x_{i}^{t} \cdot\log\left(\frac{\widehat x_{i}^{t}}{\widehat x_{i}^{t-1}}\right) - x_i^t + x_i^{t-1}\right]  + \sum_{t \in P} \sum_{i| \, p_i^t\neq 0} w_i \left[x_i^t\log\left(\frac{x_{i}^{t}}{x_{i}^{t-1}}\right) - x_i^t + x_i^{t-1}\right] \nonumber \\
	&= \sum_{t \in C} \sum_{i| \, c_i^t\neq 0} w_i\widehat x_{i}^{t} \cdot \log\left(\frac{\widehat x_{i}^{t}}{\widehat x_{i}^{t-1}}\right)  + \sum_{t \in P} \sum_{i| \, p_i^t\neq 0} w_i  \cdot x_i^t\log\left(\frac{x_{i}^{t}}{x_{i}^{t-1}}\right)- \sum_{i=1}^n  w_i(x^T_i + x^0_i) \label{eq:use_kl_limit} \\
	&  \stackrel{\eqref{eq:kkt2} \ \& \ \eqref{eq:kkt4}}{=} \quad \sum_{t\in C} y^t\sum_{i| \, c^t_i \neq 0} \left(c^t_i x_i^t + \frac{\eps}{4d^t}\right) - \sum_{t \in P} z^t\sum_{i| \, p_i^t\neq 0}p^t_i x_i^t  - \sum_{i=1}^n w_i x^T_i \nonumber \\
	& \stackrel{\eqref{eq:kkt1} \ \& \ \eqref{eq:kkt3}}{\leq} \quad \left(1+ \frac{\eps}{4}\right)\sum_{t \in C} y^t - (1+\eps)\sum_{t \in P} z^t. \label{eq:ineqyz}
\end{align}
Equation \eqref{eq:use_kl_limit} follows from \eqref{eq:kl_limit}. Inequality \eqref{eq:ineqyz} uses the fact that there are at most $d^t$ indices $i$ for which $c_i^t \neq 0$.   
\end{proof}

By a careful combination of \cref{lemmamovement} and \cref{lemmaKL}, we get the following bound on the recourse in terms of the dual objective. We write the lemma in terms of a more general $\widetilde{y}^{\tau}$ since we will need it later. For now, the reader should set $\alpha=0$, i.e. $\widetilde{y}^{\tau}=y^{\tau}$.

\begin{lemma}\label{lem:bound-recourse}
    Let $\widetilde{y}^{\tau}$ be such that
    $\sum_{\tau \in C}\widetilde{y}^{\tau} \geq \left(1 - \alpha \eps\right)\sum_{\tau \in C}y^{\tau}$ for $\alpha \in[0,\nicefrac{1}{10}]$. Then, 
 \begin{align}
\sum_{t=1}^{T}\sum_{i}w_{i}\left(x_{i}^{t}-x_{i}^{t-1}\right)_{+} & = O\left(\frac{1}{\eps}\right) \left[\sum_{t \in C} \widetilde{y}^t- \sum_{t \in P}z^t \right] \label{main-inequality}
\end{align}
\end{lemma}

\begin{proof}
Observe that
\begin{align}
\left(1+ \frac{\eps}{4}\right)\sum_{t \in C} y^t &\leq \frac{1+ \frac{\eps}{4}}{1 - \alpha \eps}\cdot \sum_{t \in C} \widetilde{y}^t \nonumber \\
&\leq \left(1+\frac{\eps}{4}\right)\cdot \left(1+2\alpha \eps\right) \cdot
\sum_{t \in C} \widetilde{y}^t \nonumber\\
&\leq \left(1+\frac{\eps}{4} + 2\alpha \eps + \frac{2\alpha \eps}{4}\right) \cdot \sum_{t \in C} \widetilde{y}^t \nonumber\\
&\leq \left(1+\frac{\eps}{2}\right)\cdot \sum_{t \in C} \widetilde{y}^t, \label{dualbound_obs}
\end{align}
where the first inequality follows from our assumption, and the second since 
$1/(1-x) \leq 1+2x$ for $x\in [0,\frac{1}{2}]$ and $\eps\leq 1$.

Summing \cref{lemmamovement} over all times $t$ and adding $\frac{2+\eps}{\eps}$ of \cref{lemmaKL}, we get:
\begin{align*}
\sum_{i}w_{i}\left(x_{i}^{t}-x_{i}^{t-1}\right)_{+} 
& \leq \sum_{t \in C} \left(1+ \frac{\eps}{4}\right)y^t +  \frac{2+\eps}{\eps} \left[\sum_{t \in C}\left(1+\frac{\eps}{4}\right)y^t- \sum_{t \in P}(1+\eps) z^t \right]  \\
& \leq \sum_{t \in C} \left(1+ \frac{\eps}{2}\right)\widetilde{y}^t +  \frac{2+\eps}{\eps} \left[\sum_{t \in C}\left(1+\frac{\eps}{2}\right)\widetilde{y}^t- \sum_{t \in P}(1+\eps) z^t \right]  \\
& = \frac{(2+\eps)(1+\eps)}{\eps}\left[\sum_{t \in C} \widetilde{y}^t- \sum_{t \in P}z^t \right]\\
&= O\left(\frac{1}{\eps}\right) \left[\sum_{t \in C} \widetilde{y}^t- \sum_{t \in P}z^t \right],
	\end{align*}
 where the second inequality is from \eqref{dualbound_obs} above.
\end{proof}

Putting things together, we can conclude with the theorem of this section.

\begin{proof}[Proof of \cref{thm:warmup}]
    By \cref{lem:bound-recourse} we can bound the total movement of the algorithm by
 \begin{align*}
	\sum_{t=1}^{T}\sum_{i}w_{i}\left(x_{i}^{t}-x_{i}^{t-1}\right)_{+} & = O\left(\frac{1}{\eps}\right) \cdot \left[\sum_{t \in C} y^t - \sum_{t \in P} z^t \right] =  O\left(\frac{\log\left(\frac{d \cdot \Delta}{\eps}\right)}{\eps}\right) \cdot \left[\sum_{t \in C} \bar{y}^t - \sum_{t \in P} \bar{z}^t \right].
 \end{align*}
 Since $(\overline y, \overline z, \overline r)$ is a feasible dual solution, the theorem follows from weak duality.
\end{proof}

\subsection{Removing the $\Delta$: An $O\left(\nicefrac{1}{\eps}\log \left(\nicefrac{d}{\eps}\right)\right)$ competitive bound}
\label{sec:removing_delta}
We move to proving the more refined bound. 
The idea is once again to fit a dual to $\overline{\mathcal{D}}$, but this time we need to construct our dual solution more delicately. For one, we need to overcome the following significant barrier: the dual $y$ can be computed monotonically online, however one can show that no monotone dual solution can avoid a dependence on $\Delta$, see \cite{BN09}. Nevertheless, we show that it is possible to \emph{decrease} some coordinates of the dual carefully, to achieve the following bound.

\begin{theorem}[$\Delta$-free Bound]
\label{thm:sparse}
The total recourse of the algorithm is at most $O\left(\nicefrac{1}{\eps}\cdot \log\left(\nicefrac{d}{\eps}\right)\right)$ times the objective of $(\overline{\mathcal{D}})$.
\end{theorem}

We will construct a vector $\widetilde y$ with the following properties.

\begin{lemma}\label{lem:sp_dualok}
Given $y$, there is a vector $\widetilde{y}$ such that
\begin{align}
    &\forall 1 \leq s \leq t \leq T, \ \forall i  &  \sum_{\substack{\tau \in C: \\ \tau \in [s, t]}}c_i^{\tau} \widetilde{y}^{\tau} - \sum_{\substack{\tau \in P: \\ \tau \in[s, t]}}p^{\tau}_i z^{\tau} &\leq w_i \cdot \log\left(1 + \frac{40 d^2}{\eps^2}\right) \label{ineq1}\\
    & \text{and} & \sum_{\tau \in C}\widetilde{y}^{\tau} &\geq \left(1 - \frac{\eps}{10}\right)\sum_{\tau \in C}y^{\tau}. \label{ineq2}
\end{align}
\end{lemma}

Once we have this, we can define our dual solution as follows. Let $A= \log\left(1 + \frac{40 d^2}{\eps^2}\right)$. We set,
\begin{align*}
    \forall s\in C && \bar{y}^{s} &= \widetilde{y}^{s}/A, \\
    \forall s\in P && \bar{z}^{s} & = z^{s}/A, \\
    && \overline{r}_{i}^{T+1} &= 0,\\
    \forall s \in [T] && \overline{r}_{i}^{s} &= \frac{1}{A} \cdot \max\left(0, \ \max_{t \geq s}\left\{\sum_{\substack{\tau \in C: \\ \tau \in [s, t]}}c_i^{\tau} \widetilde{y}^{\tau} - \sum_{\substack{\tau \in P: \\ \tau \in [s, t]}}p^{\tau}_i z^{\tau}\right\}\right).
\end{align*}
Before proving \cref{lem:sp_dualok}, let us understand why this dual solution is feasible and implies \cref{thm:sparse}.

\begin{lemma}
    The variables $(\bar{y}, \bar{z}, \bar{r})$ are a feasible dual solution to  $(\overline{\mathcal{D}})$.
\end{lemma}

\begin{proof}
We start by proving the feasibility of the dual solution.
The constraints $\overline{r}_{i}^s \geq 0$ hold by construction. For a sequence $a_1, a_2, \ldots, a_T$ of positive and negative numbers is holds that for any $s \in [T]$:
\begin{align*}
\max\left(0, \ \max_{t | \, s\leq t \leq T}\left\{\sum_{\tau \in [s, t]}a_\tau\right\}\right)  
& \geq \max\left(a_s, \ a_s+\max_{t | \, s+1\leq t \leq T}\left\{\sum_{\tau \in [s+1, t]}a_\tau\right\}\right) \\
&= a_s + \max\left(0, \ \max_{t | \, s+1\leq t \leq T}\left\{\sum_{\tau \in [s+1, t]}a_\tau\right\}\right) 
\end{align*}
By this observation and plugging $a_t=c_i^{\tau} \overline{y}^{\tau}$, if $t\in C$ or $a_t= -p^{\tau}_i \bar{z}^{\tau}$ for $t\in P$, we immediately get that $\overline{r}_{i}^{s} \geq c_i^{\tau} \overline{y}^{\tau} + \overline{r}_{i}^{s+1}$, or similarly $\overline{r}_{i}^{s} \geq -p^{\tau}_i \bar{z}^{\tau} + \overline{r}_{i}^{s+1}$.

Finally, by inequality \eqref{ineq1} of \cref{lem:sp_dualok} we get:
\begin{align*}
\overline{r}_{i}^{s} &= \frac{1}{A} \cdot \max\left(0, \ \max_{t \geq s}\left\{\sum_{\substack{\tau \in C: \\ \tau \in [s, t]}}c_i^{\tau} \widetilde{y}^{\tau} - \sum_{\substack{\tau \in P: \\ \tau \in [s, t]}}p^{\tau}_i z^{\tau}\right\}\right) \leq \frac{1}{A} \cdot w_i \cdot \log\left(1 + \frac{40 d^2}{\eps^2}\right) \leq w_i,
\end{align*}
and the dual solution is feasible.
\end{proof}
   Our main theorem now follows by putting everything together.
 \begin{proof}[Proof of \cref{thm:sparse}]
By \eqref{ineq2} of \cref{lem:sp_dualok} we have $\sum_{\tau \in C}\widetilde{y}^{\tau} \geq \left(1 - \frac{\eps}{10}\right)\sum_{\tau \in C}y^{\tau}$, so we can reuse \cref{lem:bound-recourse} to bound the total recourse as
 \begin{align*}
	\sum_{t=1}^{T}\sum_{i}w_{i}\left(x_{i}^{t}-x_{i}^{t-1}\right)_{+} & = O\left(\frac{1}{\eps}\right) \cdot \left[\sum_{t \in C} \widetilde{y}^t - \sum_{t \in P} z^t \right] =  O\left(\frac{\log\left(\frac{d}{\eps}\right)}{\eps}\right) \cdot \left[\sum_{t \in C} \bar{y}^t - \sum_{t \in P} \bar{z}^t \right]. 
 \end{align*}
 Since $(\overline y, \overline z, \overline r)$ is a feasible dual solution, the theorem follows from weak duality.
 \end{proof}

The remainder of the section is devoted to proving \cref{lem:sp_dualok}. First we need an auxiliary lemma.

\begin{lemma}
    \label{lem:subsetlem}
    For any $s$ and $t$ such that $1\leq s\leq t\leq T$, any $i\in [n]$, and any $S \subseteq C \cap [s,t]$, define $c_i^{\max}(S) = \max_{\tau \in S} \{c_i^\tau\}$. Then
	\begin{align*}
		\frac{1}{w_i}\left(\sum_{\substack{\tau \in S: \\ \tau \in[s,t]}}c^\tau_iy^\tau - \sum_{\substack{\tau \in P: \\ \tau \in [s,t]}} p_i^\tau z^\tau\right) & \leq \log\left(1+\frac{4 d \cdot c_i^{\max}(S)}{\eps} \cdot  x_i^{t}\right).
	\end{align*}
\end{lemma}

\begin{proof}
We have that
\begin{align}
\frac{1}{w_i}\left(\sum_{\substack{\tau \in S: \\ \tau \in[s, t]}}c^\tau_i y^\tau- \sum_{\substack{\tau\in P: \\ \tau \in[s, t]}} p_i^\tau z^\tau\right) &= \sum_{\substack{\tau\in S: \\ \tau \in[s, t], \, c^\tau_i\neq 0}} \log\left(\frac{x_{i}^{\tau}+ \frac{\eps}{4 d^t \cdot c^t_i}}{x_{i}^{\tau-1}+ \frac{\eps}{4 d^t \cdot c^t_i}}\right) + \sum_{\substack{\tau\in P: \\ \tau \in[s, t]}} \log\left(\frac{x_{i}^{\tau}}{x_{i}^{\tau-1}}\right) \label{eq:sp_usekkt} \\
& \leq \sum_{\tau=s}^{t} \log\left(\frac{x_i^\tau + \frac{\eps}{4 d \cdot c_i^{\max}(S)}}{x_i^{\tau-1}+ \frac{\eps}{4d \cdot c_i^{\max}(S)}}\right) \label{eq:sp_usekkt2}\\
&= \log\left(\frac{x_i^{t}+ \frac{\eps}{4 d \cdot c_i^{\max}(S)}}{x_i^{s-1}+ \frac{\eps}{4d \cdot c_i^{\max}(S)}}\right) \nonumber \\
&\leq \log\left(1+\frac{4 d \cdot c_i^{\max}(S)}{\eps} \cdot  x_i^{t}\right).\nonumber
\end{align}

Equality \eqref{eq:sp_usekkt} follows from \eqref{eq:kkt2} and \eqref{eq:kkt4}.
To see why inequality \eqref{eq:sp_usekkt2} follows, note that for any $i$ and $\tau\in C$, it holds that $x_i^{\tau} \geq x_i^{\tau-1}$ and therefore, for any $a\geq b$ we have $\frac{x_i^{\tau}+ a}{x_i^{\tau-1}+ a}\leq \frac{x_i^{\tau}+ b}{x_i^{\tau-1}+ b}$. For any $\tau\in P$, $x_i^{\tau} \leq x_i^{\tau-1}$ and therefore, for any $a\geq 0$, we get $\frac{x_i^{\tau}}{x_i^{\tau-1}}\leq \frac{x_i^{\tau}+ a}{x_i^{\tau-1}+ a}$. Finally, for any $\tau \in (C \cap [s,t]) \backslash S$, we have that $x_i^{\tau} \geq x_i^{\tau-1}$ and therefore the corresponding summand is nonnegative.
\end{proof}

With this, we are finally ready to prove the lemma.

\begin{proof}[Proof of \cref{lem:sp_dualok}]
We obtain $\widetilde{y}$ by inductively \emph{decreasing} coordinates of the solution $y$ produced by the algorithm. We prove by induction on $\ell$ that inequalities \eqref{ineq1} and \eqref{ineq2} hold until time $\ell$, i.e.,
\begin{align}
    &\forall 1 \leq s \leq t \leq \ell, \ \forall i  &  \sum_{\substack{\tau \in C: \\ \tau \in [s, t]}}c_i^{\tau} \widetilde{y}^{\tau} - \sum_{\substack{\tau \in P: \\ \tau \in[s, t]}}p^{\tau}_i z^{\tau} &\leq w_i \cdot \log\left(1 + \frac{40d^2}{\eps^2}\right) \label{ih1} \\
    & \text{and} & \sum_{\substack{\tau \in C \\ \tau \leq \ell}}\widetilde{y}^{\tau} &\geq \left(1 - \frac{\eps}{10}\right)\sum_{\substack{\tau \in C \\ \tau \leq \ell}}y^{\tau}. \label{ih2}
\end{align}

The base of the induction is $\ell=0$ which holds trivially. Now suppose by induction that
we have an assignment of $\widetilde{y}^\tau$ for all $\tau\leq \ell -1$ such that inequalities \eqref{ih1} and \eqref{ih2} hold until time $\ell-1$ (and $\widetilde{y}^\tau =0$ for $\tau \geq \ell$). We show how to create an updated assignment $\widetilde{y}^\tau_{new}$ such that these inequalities also hold for time $\ell$.  

If time $\ell \in P$, the inequalities hold already and there is no need to change $\widetilde{y}^\tau$. Therefore, assume that  $\ell\in C$. We construct $\widetilde{y}_{new}$ as follows.

\begin{algorithm}[H]
\caption{Constructing $\widetilde{y}_{new}$}
          \label{alg:sp_fixy}
	\begin{algorithmic}[1]
	\For{$i = 1, \ldots, n$ \label{line:iforloopstart}}
        \State Initialize \begin{spacing}{0.3}
        \begin{flalign*}
        R_i &\gets \left\{\tau \in [\ell-1] \ \middle| \ c_i^{\tau}\geq \frac{10 \cdot d^\ell \cdot c_i^\ell}{\eps}, \ \widetilde{y}^{\tau} > 0 \right\},&& \text{// candidates for decrease} \\
        B &\gets c_i^\ell \cdot y^\ell. && \text{// budget} 
        \end{flalign*} 
        \end{spacing}
        \While{$B > 0$ and $R_i \neq \emptyset$}
        \State Let $\tau^*$ be the latest time in $R_i$.
        
        \State Update\begin{spacing}{0.3}
        \begin{flalign*}
            B &\gets B - c_i^{\tau^*} \Delta \widetilde{y}^{\tau^*}_i. &
        \end{flalign*} 
        \end{spacing}
        \For{every $\tau \in [\ell - 1]$}
        \State Update \begin{spacing}{0.3}
        \begin{flalign*}
            \Delta\widetilde{y}^\tau_i &\gets \min\{\widetilde{y}^\tau, B / c_i^{\tau*}\}. &
        \end{flalign*} 
        \end{spacing}
        \State If  $\Delta\widetilde{y}^\tau_i= \widetilde{y}^\tau$, then also update $R_i \gets R_i \setminus\{\tau^*\}$.
        \EndFor
        \EndWhile
        \EndFor 
        \State Finally, set \vspace{-0.1in}\begin{spacing}{0.5}
            \begin{flalign*}
                \widetilde{y}^\ell_{new} &  \gets y^\ell,
                \intertext{and for all $t \in [\ell-1]$}
                \widetilde{y}^\tau_{new} & \gets \widetilde{y}^{\tau} - \max_{i | \, c_i^t\neq 0}\{\Delta \widetilde{y}^{\tau}_i\}. &
            \end{flalign*}
            \end{spacing}
	\end{algorithmic}
\end{algorithm}

By construction, at the end of while loop which always terminates, either $R_i=\emptyset$, or the total budget spent is $\Delta B=c_i^\ell y_i^\ell= \sum_{\tau \leq \ell-1}c_i^{\tau}\Delta \widetilde{y}^\tau_i$. Throughout the execution, we always maintain that 
\begin{align*}
    c_i^\ell y_i^\ell- \sum_{\tau\leq \ell-1}c_i^{\tau}\Delta \widetilde{y}^\tau\geq 0
\end{align*}
and for all $\tau \in [\ell-1]$, that $\Delta \widetilde{y}^{\tau}_i\leq \widetilde{y}^{\tau}$.

Since $c_i^\ell y_i^\ell - \sum_{\tau}c_i^{\tau}\Delta \widetilde{y}_i^\tau \geq 0$ and for each $\tau\in R_i$, by definition $c_i^{\tau}\geq 10 \cdot d^\ell c_i^\ell / \eps$, we get that 
\begin{align}\sum_{\tau\leq \ell-1}\Delta \widetilde{y}^{\tau}_i \leq   \frac{\epsilon}{10 \cdot d^t}\cdot \sum_{\tau\leq \ell-1}\frac{c_i^{\tau} \cdot \Delta \widetilde{y}^{\tau}_i}{c_i^t} \leq \frac{\epsilon}{10 \cdot d^t} \cdot y^\ell. \label{delta_small} \end{align}
Thus, even after the decreases of \cref{alg:sp_fixy}, the total mass we add to $\widetilde y$ is at least
\begin{align*}
(\widetilde{y}^\ell_{new}- \widetilde{y}^\ell) + \sum_{\tau\leq t-1}(\widetilde{y}^\tau_{new}- \widetilde{y}^\tau)  & = 
 y^\ell - \sum_{\tau\leq \ell-1}\max_{i | \, c_i^t\neq 0}\{\Delta \widetilde{y}^{\tau}_i\} \\
 & \geq y^\ell - \sum_{i| \, c_i^{t}\neq 0}\sum_{\tau\leq \ell-1}\Delta \widetilde{y}^{\tau}_i \\
 &\geq y^\ell - \frac{\eps}{10 \cdot d^\ell}\sum_{i| \, c_i^{\ell}\neq 0}y^\ell \\
 &= \left(1-\frac{\eps}{10}\right)y^\ell. 
\end{align*}
Above, the first equality follows from the fact that we defined $\widetilde{y}^\ell$ to be $0$. The second inequality came from \eqref{delta_small}, and the third from the fact that at most $d^\ell$ coordinates of $c^\ell$ are nonzero. Combining this with the inductive hypothesis \eqref{ih2},
\begin{align*}
\sum_{\tau \leq \ell} \widetilde{y}_{new}^\tau = (\widetilde{y}^\ell_{new}- \widetilde{y}^\ell) + \sum_{\tau\leq \ell-1}(\widetilde{y}^\tau_{new}- \widetilde{y}^\tau) + \sum_{\tau \leq \ell-1} \widetilde{y}^\tau \geq \left(1-\frac{\eps}{10}\right) \sum_{\substack{\tau \in C \\ \tau \leq \ell}} y^\ell,
\end{align*}
and hence inequality \eqref{ih2} is preserved.

It remains to prove that for any $i\in [n]$ and $s\leq t = \ell$, inequality \eqref{ih1} holds. Fix an index $i$ and let $\tau_i^{\min}$ be the earliest time $\tau$ for which $\Delta\widetilde{y}_i^\tau >0$. We split into cases.

\textbf{Case 1: $\mathbf{c_i^\ell=0}$}. In this case
\[\sum_{\substack{\tau \in C: \\ \tau \in [s, \ell]}}c_i^{\tau} \widetilde{y}^{\tau}_{new} - \sum_{\substack{\tau \in P: \\ \tau \in[s, \ell]}}p^{\tau}_i z^{\tau} = \sum_{\substack{\tau \in C: \\ \tau \in [s, \ell-1]}}c_i^{\tau} \widetilde{y}^{\tau}_{new} - \sum_{\substack{\tau \in P: \\ \tau \in[s, \ell-1]}}p^{\tau}_i z^{\tau} \leq \sum_{\substack{\tau \in C: \\ \tau \in [s, \ell-1]}}c_i^{\tau} \widetilde{y}^{\tau} - \sum_{\substack{\tau \in P: \\ \tau \in[s, \ell-1]}}p^{\tau}_i z^{\tau}.\]
and inequality \eqref{ih1} holds by the induction hypothesis. Assume in the remaining cases that $c_i^\ell > 0$.

\textbf{Case 2: $\mathbf{R_i\neq \emptyset}$ at the end of the loop and $\mathbf{s \leq \tau_i^{\min}}$.} In this case, the budget $B$ was completely spent, i.e.
\[c_i^\ell y^\ell= \sum_{\substack{\tau \in C \\ \tau \in [\tau_i^{\min}, \, \ell-1]}}c_i^{\tau}\Delta \widetilde{y}^\tau_i,\] 
which means that
\begin{align*}
  \sum_{\substack{\tau \in C: \\ \tau \in [\tau_i^{\min}, \, \ell]}}c_i^{\tau} \widetilde{y}^{\tau}_{new} - \sum_{\substack{\tau \in P: \\ \tau \in[\tau_i^{\min}, \, \ell]}}p^{\tau}_i z^{\tau} \leq \sum_{\substack{\tau \in C: \\ \tau \in [\tau_i^{\min}, \, \ell-1]}}c_i^{\tau} \widetilde{y}^{\tau} - \sum_{\substack{\tau \in P: \\ \tau \in[\tau_i^{\min}, \, \ell-1]}}p^{\tau}_i z^{\tau}.
\end{align*}
Then, for any $s\leq \tau_i^{\min}$, 
\begin{align*}
\sum_{\substack{\tau \in C: \\ \tau \in [s, \ell]}}c_i^{\tau} \widetilde{y}^{\tau}_{new} - \sum_{\substack{\tau \in P: \\ \tau \in[s, \ell]}}p^{\tau}_i z^{\tau}
& \leq \sum_{\substack{\tau \in C: \\ \tau \in [s, \ell-1]}}c_i^{\tau} \widetilde{y}^{\tau} - \sum_{\substack{\tau \in P: \\ \tau \in[s, \ell-1]}}p^{\tau}_i z^{\tau},
\end{align*}
and again, inequality \eqref{ih1} holds by the induction hypothesis.

\textbf{Case 3: $\mathbf{R_i = \emptyset}$ at the end of the loop or $\mathbf{s > \tau_i^{\min}}$.} In this case $[s,\ell] \cap R_i = \emptyset$ and therefore, for each $\tau\in [s,\ell]$ either $\widetilde{y}^{\tau}_{new}=0$, or $c_i^\tau \leq 10 \cdot d^\ell c_i^\ell/\eps$. Define
\begin{align*}
    S&= C \cap \{\tau\in[s,\ell], \widetilde{y}^\tau>0\}, \\
    c_i^{\max}(S) &= \max_{\tau \in S} \{c_i^\tau\} \leq 10 \cdot d^\ell c_i^\ell / \eps.
\end{align*}
Applying \cref{lem:subsetlem} with the set $S$ above, we get that:
\begin{align*}
  \frac{1}{w_i}\left(\sum_{\substack{\tau \in C \\ \tau  \in[s,t]}}c^\tau_i \widetilde{y}^\tau - \sum_{\substack{\tau \in P: \\ \tau \in [s,t]}} p_i^\tau z^\tau\right) &=  \frac{1}{w_i}\left(\sum_{\substack{\tau \in S}}c^\tau_i y^\tau - \sum_{\substack{\tau \in P: \\ \tau \in [s,t]}} p_i^\tau z^\tau\right)\\
  & \leq  \log\left(1+\frac{4 d \cdot c_i^{\max}(S)}{\eps} \cdot  x_i^{t}\right)\\
  & \leq \log\left(1+\frac{4 d}{\eps} \cdot \frac{10 \cdot d^\ell c_i^\ell}{\eps} \cdot  \frac{1}{c_i^\ell}\right) \\
  &\leq \log\left(1+\frac{40 d^2}{\eps^2}\right).
\end{align*}
The second inequality follows from the definition of $c_i^{\max}(S)$. The last inequality follows from the fact that $x_i^\ell\leq \nicefrac{1}{c_i^\ell}$, since otherwise the constraint at time $\ell$ is already satisfied at time $\ell-1$. This concludes the proof.
\end{proof}

\section{Rounding}

\label{sec:rounding}

In this section we study several applications of the chasing positive bodies problem. For each of these problems we show that a suitable fractional version of it fits into our framework. We then show how to round the fractional solution to an integral solution. 

\subsection{Dynamic Set Cover}

\label{sec:setcov}

In the set cover problem we are given a universe of $n$ elements $U$ and a collection of $m$ sets $S_1, \ldots, S_m \subseteq U$, where each set $S_i$ has a cost $c(S_i) \geq 0$. The goal is to find a minimum cost collection of sets that covers all elements. 

In the dynamic setting, at each time-step $t\in[T]$, an element gets inserted into/deleted from the universe $U$. Let $U^t$ denote the state of $U$ at time $t$,\footnote{We assume that each element in $U^t$ is covered by at least one set in $\{S_1, \ldots, S_m\}$. This ensures the existence of a valid set cover.} and let $\opt^t$ denote the cost of the optimal {\em fractional} set cover for $U^t$.  We wish to maintain a set cover $\mathcal{S}^t$ of $U^t$ whose cost $c(\mathcal{S}^t) := \sum_{S \in \mathcal{S}^t} c(S)$ is always within a factor $\beta \geq 1$ of $\opt^t$, such that the total recourse, $\sum_{t \in [T]} |\mathcal{S}^t \oplus \mathcal{S}^{t-1}|$, is as small as possible. The fractional version of the dynamic set cover problem fits naturally into our framework 
\begin{align*}
    K_t & := \left\{x^t\in \mathbb{R}^m_+ \ \middle| \begin{array}{cc}
         \sum_{\substack{i \in [m] \, : \\ u \in S_i}} x^t_i \geq 1 & \forall u\in U^t, \\ \\
         \sum_{i \in [m]} c(S_i) \cdot x^t_i  \leq  \beta \cdot \opt^t & 
    \end{array}\ \right\}.
\end{align*}
\begin{restatable}[Set Cover Rounding]{theorem}{setcovround}
\label{setcover_rounding}
    For any fractional solution $x^t$ to dynamic set cover, we have:
\begin{enumerate}[label={(\arabic*)}]
    \item A deterministic dynamic rounding scheme such that
\begin{align}
    \sum_{t \in [T]}|\mathcal{S}^t \oplus \mathcal{S}^{t-1}| & \leq O(f)\cdot \sum_{t \in [T]}\|x^t- x^{t-1}\|_1, \nonumber \\
    \sum_{S\in \mathcal{S}^{t}} c(S)  & \leq O(f)\cdot \sum_{i \in [m]} c(S_i) \cdot x^t_i & \forall t \in [T]. \nonumber
\end{align}   
Here, $f$ is an upper bound on the maximum frequency of any element at any time, i.e., 
\[f \geq \max_{u\in U^t} |\{i \in [m]  :  u\in S_i\}|   \ \  \forall t \in [T].\]
\item For every  $\alpha \geq 1$, a  
randomized dynamic rounding scheme such that
\begin{align}
    \expect*{\sum_{t \in [T]}\left\vert\mathcal{S}^t \oplus \mathcal{S}^{t-1}\right\vert} & \leq (\alpha \log n)\cdot \sum_{t \in [T]}\|x^t- x^{t-1}\|_1 + O\left(\frac{T}{n^\alpha}\right),
    \label{line:sc_rounding_recourse}\\
    \expect*{\sum_{S\in \mathcal{S}^{t}} c(S)} & \leq O(\alpha \log n)\cdot \sum_{i \in [m]} c(S_i) \cdot x^t_i  & \forall t \in [T].
    \label{line:sc_rounding_apx}
\end{align}   
Here,  $n \geq \max_t |U^t|$ is an upper bound on the  number of elements in the universe at any time. 
\end{enumerate}
\end{restatable}
The proof, which we defer to \cref{appendix_a}, is an adaptation of standard offline rounding techniques to the dynamic setting. Note that in this problem the row sparsity of the covering constraints of $K_t$ is at most $f \geq \max_{u\in U^t} |\{i \in [m]  :  u\in S_i\}|$.  Hence, by setting $\eps=1$ in  \cref{thm:main},  we can maintain a fractional solution $x^t$ that satisfies the covering constraints, at any time $t$ has cost at most $O(\beta) \cdot \opt^t$, and has a   $O(\log f)$-competitive total recourse with respect to any offline algorithm maintaining a fractional solution of cost  $\leq \beta \cdot \opt^t$. Also note that if we set $\alpha$ to be a constant, then the additive term in \eqref{line:sc_rounding_recourse} becomes $1/\poly(n)$ per time-step. Combining this with \cref{setcover_rounding} we get:

\begin{corollary}\label{cor:set cover}
For any $\beta \geq 1$ and any update-sequence $\sigma$ consisting of $T$ updates, there exist:
\begin{enumerate}[label={(\arabic*)}]
    \item A deterministic dynamic algorithm that maintains a set cover $\mathcal{S}^t$ with cost $O(\beta \cdot f \cdot \opt^t)$ at all times, whose total recourse is at most $O(f\log f)\cdot \optr^\beta(\sigma)$.
    \item A randomized dynamic algorithm that maintains a set cover $\mathcal{S}^t$ with expected cost $O(\beta \cdot \log n \cdot \opt^t)$ at all times, whose expected total recourse is at most $O(\log f \cdot \log n)\cdot \optr^\beta(\sigma) + \frac{T}{\poly(n)}$.
\end{enumerate}
Here, $\optr^\beta(\sigma)$ is the optimal total recourse for any {\em offline} algorithm which maintains a {\em fractional} set cover of cost at most $\beta \cdot \opt^t$ throughout the update-sequence $\sigma$.
\end{corollary}

\medskip

We observe that we can easily extend our analysis to allow for other forms of updates, such as changing the cost of a set.

\subsection{Dynamic Bipartite Matching}

\label{sec:matching}

In the maximum  bipartite matching problem, we are given a bipartite graph $G = (V, E)$ with $|V|=n$ nodes. 
A matching $M\subseteq E$ is subset of edges that do not share any common endpoint,  and the goal is to find a matching $M$ of maximum size.
In the dynamic setting, at each time-step $t\in [T]$ an edge gets inserted into/deleted from the set $E$.  Let $E^t$ denote the state of $E$ at time $t$, and let $\opt^t$ denote the value of the maximum fractional matching at time $t$. We wish to maintain a matching $M^t$ in $G^t = (V, E^t)$ whose size is always within a constant factor $\beta \leq 1$ of $\opt^t$, such that the total recourse, $\sum_{t \in [T]} |M^t \oplus M^{t-1}|$, is as small as possible. The fractional version of the dynamic maximum matching problem fits naturally into our framework with the following 
\begin{align}
\label{eq:kt:matching}
    K_t & := \left\{x^t\in \mathbb{R}^{\binom{|V|}{2}}_+ \ \middle| \ \begin{array}{cc}
        \sum_{e \in \partial(v)} x^t_{e}  \leq 1 & \forall v \in V, \\
        x_e^t = 0 & \forall e \in \binom{V}{2} \setminus E^t, \\
        \sum_{e \in E^t}  x_e^t  \geq  \beta \cdot \textsc{Opt}^t. &
    \end{array} \ \right\}.
\end{align}
Note that the vector $x^t$ has an entry for every {\em potential edge} (i.e., unordered pair of nodes) $e \in \binom{V}{2}$. If a potential edge $e$ is currently not present in the graph, then we set $x^t_e = 0$. The notation $\partial(v)$ denotes the set of all potential edges with one endpoint in $v \in V$.

The main result in this section is summarized  below. The proof of \cref{th:matching:rounding}, which we defer to \cref{app:matching}, follows from an adaptation of a standard rounding technique~\cite{ArarCCSW18} to our setting.

\begin{restatable}[Maximum Matching Rounding]{theorem}{maxmatchround}
\label{th:matching:rounding}
For any fractional solution $x^t$ to dynamic bipartite matching, and any $\delta \in (0,1)$, $\alpha\geq 1$, there exists a randomized dynamic 
 rounding scheme such that:
\begin{align}
    \expect*{\sum_{t \in [T]} \left\vert M^t \oplus M^{t-1} \right\vert} & \leq O\left(\frac{\alpha \log n}{\delta^3}\right) \cdot \sum_{t \in [T]}\|x^t- x^{t-1}\|_1 +  O\left(\frac{T}{\delta \cdot n^{\alpha}}\right). \label{eq:matching:rounding:1}\\
    \expect*{\left| M^{t} \right|} & \geq (1-\delta)\cdot \sum_{e\in E^t} x^t_e  & \forall t \in [T].\label{eq:matching:rounding:2}
\end{align}    
\end{restatable}

Note that in this problem the row-sparsity of the covering constraint in $K_t$ is at most $\binom{n}{2} = \Theta(n^2)$. Thus, using the online algorithm from \cref{thm:main}, we can maintain a fractional matching $x^t$ that satisfies $\sum_{e \in E^t}  x^t_e \geq (1-\epsilon) \beta \cdot \opt^t$ at all times $t$, and has a $O((\nicefrac{1}{\epsilon}) \cdot \log (\nicefrac{n}{\epsilon}))$-competitive total recourse with respect to any offline algorithm which maintains a fractional matching of value $\geq \beta \cdot \opt^t$. Setting $\delta = \epsilon$ and $\alpha$ to be some large constant, and combining this with \cref{th:matching:rounding}, we now get:
\begin{corollary}
\label{cor:matching:rounding} For any $\beta, \epsilon \in (0, 1]$, and any update-sequence $\sigma$ with $T$ updates, there exists a randomized dynamic algorithm that maintains a matching $M^t$ of expected size at least $(1-\epsilon) \beta \cdot \opt^t$ at all times, with expected total recourse $O((\nicefrac{1}{\epsilon^4}) \cdot \log n \cdot \log (\nicefrac{n}{\epsilon})) \cdot  \optr^{\beta}(\sigma) +  \frac{T}{\epsilon \cdot \poly(n)}$. Here, $\optr^{\beta}(\sigma)$ is the optimal total recourse for any {\em offline} algorithm which maintains a fractional matching of value at least $\beta \cdot \opt^t$ throughout the update-sequence $\sigma$.
\end{corollary}

We remark that it is straightforward to extend \cref{th:matching:rounding} and \cref{cor:matching:rounding} to weighted bipartite graphs, at the cost of incurring an exponential in $(1/\epsilon)$ factor overhead in recourse, using a well-known reduction due to Bernstein et al. (see Theorem 3.1 in~\cite{BernsteinDL21}). We can also easily extend our analysis to also allow for vertex updates as well.

\subsection{Dynamic Load Balancing on Unrelated Machines}

\label{sec:load_balance}

In the load balancing problem, we are given a set $M$ of $m$ {\em machines}, and a set $J$ of  {\em jobs}. Each job $j \in J$ can only be processed by a nonempty subset $M(j) \subseteq M$ of machines. If we assign job $j$ to machine $i \in M(j)$, then the machine incurs a {\em load} of $p_{ij}$. Our goal is to come up with a valid {\em assignment} $\psi : J \rightarrow M$ of jobs to machines so as  to minimize the ``makespan'', which is defined to be the maximum load on any machine and is given by $\texttt{Obj}(\psi) := \max_{i \in M} \left( \sum_{j \in J : \psi(j) = i} p_{ij} \right)$.

In the dynamic setting, at each time-step $t$ a job gets inserted into/deleted from the set $J$. (Our analysis also works for other types of updates, such as changing the set $M(j)$ for a job $j$.) Let $J^t$ denote the set of jobs at time $t$. Let $\opt^t_{int}$ denote the optimal {\em integral} makespan at time $t$.  We wish to maintain a valid assignment $\psi^t : J^t \rightarrow M$ such that $\texttt{Obj}(\psi^t) \leq \beta 
\cdot \opt^t_{int}$  for some $\beta \geq 1$. We incur a ``recourse'' of one whenever  a job gets reassigned from one machine to another, or whenever a job gets inserted into/deleted from the set $J$. Overloading the notation $\oplus$, which denotes the symmetric difference between two sets, we define the recourse incurred  at time $t$ to be $\psi^{t-1} \oplus \psi^t := |J^t \oplus J^{t-1}| + | \{ j \in J^t \cap J^{t-1} : \psi^t(j) \neq \psi^{t-1}(j)\}|$. We wish to ensure that the total recourse, given by $\sum_{t \in [T]} \psi^t \oplus \psi^{t-1}$, is as small as possible. Let $\mathcal{J}$ denote the collection of all possible jobs that can ever be present, so that we have $J^t \subseteq \mathcal{J}$ for all $t$, and let $n := |\mathcal{J}|$. The fractional version of this dynamic  problem fits into our framework, where the set $K_t$ consists of all $x^t \in \mathbb{R}_+^{|\mathcal{J}| \times |M|}$ that satisfy the following constraints:
\begin{eqnarray}
\label{eq:kt:makespan:1}
\sum_{i \in M} x^t_{ij} & \geq & 1 \qquad \qquad \ \ \ \forall j \in J^t. \\
x^t_{ij} & = & 0 \qquad \qquad \ \ \ \forall i \in M, j \in J^t \text{ with }  p_{ij} > \opt^t_{int}. \label{eq:kt:makespan:2} \\
x^t_{ij} & = & 0 \qquad \qquad \ \ \ \forall i \in M, j \in \mathcal{J} \setminus J^t. \label{eq:kt:makespan:3} \\
x_{ij}^t & = & 0 \qquad \qquad \ \ \  \forall j \in \mathcal{J}, i \in M \setminus M(j). \label{eq:kt:makespan:new} \\
\sum_{j \in J^t} p_{ij} \cdot x^t_{ij} & \leq & \beta \cdot \opt^t_{int} \ \ \ \,   \forall i \in M. \label{eq:kt:makespan:4}
\end{eqnarray}
 The main result in this section is summarized below. The proof of \cref{th:kt:makespan} follows almost immediately from a recent result by Krishnaswamy et al.~\cite{KLS23}.
 Parts 1 and 2 of the theorem below follow from Appendix B and Section 4 of \cite{KLS23}, respectively.\footnote{The final result on load balancing of \cite{KLS23} handles only job arrivals because of the limitation of their online fractional algorithm. However, their rounding schemes work in the fully dynamic setting.}

\begin{restatable}[Load Balancing Rounding by \cite{KLS23}]{theorem}{loadbalanceround}
\label{th:kt:makespan}
Consider any fractional solution $x^t$ to the dynamic load balancing problem on unrelated machines. Then we have:
\begin{enumerate}[label={(\arabic*)}]
\item
For any $\delta > 0$,  a deterministic dynamic rounding scheme such that:
\begin{eqnarray}
\label{eq:th:kt:makespan:1}
\sum_{t \in [T]}  \psi^t \oplus \psi^{t-1}& \leq &  O\left(\frac{1}{\delta^4} \log \frac{1}{\delta} \cdot \log n \right) \cdot \sum_{t \in [T]} \norm{x^t - x^{t-1}}_1, \\
\label{eq:th:kt:makespan:2}
{\tt Obj}(\psi^t) & \leq & (2+\delta)  \beta \cdot \opt^t_{int} \qquad  \forall t \in [T].
\end{eqnarray} 
\item A randomized dynamic rounding scheme such that:
\begin{eqnarray}
\label{eq:th:kt:makespan:20}
\expect*{\sum_{t \in [T]}  \psi^t \oplus \psi^{t-1} }  & \leq & O(1) \cdot \sum_{t \in [T]} \norm{x^t - x^{t-1}}_1 +  \frac{T}{\poly(n)}, \\
\label{eq:th:kt:makepsan:21}
\expect*{{\tt Obj}(\psi^t)} & = & O\left( \frac{\log \log n}{\log \log \log n} \right) \cdot \opt^t_{int} \qquad \forall t \in [T].
\end{eqnarray}
\end{enumerate}
\end{restatable}

The row-sparsity of the covering constraints in $K_t$ is  $r := \max_{j \in \mathcal{J}} |M(j)|$. Thus, by \cref{thm:main}, we can maintain a fractional solution $x^t$ that satisfies the four constraints \eqref{eq:kt:makespan:1}-\eqref{eq:kt:makespan:new} as well as the constraint $\sum_{j \in J^t} p_{ij} \cdot x_{ij}^t \leq (1+\epsilon)  \beta \cdot \opt^t_{int}$ at all times $t$. Furthermore this solution has total recourse that is $O((\nicefrac{1}{\epsilon}) \log (\nicefrac{r}{\epsilon}))$-competitive  with respect to any offline algorithm which maintains a feasible fractional solution to all five constraints~\eqref{eq:kt:makespan:1}-\eqref{eq:kt:makespan:4}. Setting $\delta = \epsilon$ and combining this with  \cref{th:kt:makespan}, we get:

\begin{corollary}
\label{cor:kt:makespan} For any $\beta \geq 1$, $\epsilon \in (0, 1]$, and any update-sequence $\sigma$ with $T$ updates, there exist:
\begin{enumerate}[label={(\arabic*)}]
\item A deterministic dynamic algorithm that maintains an assignment $\psi^t$ with makespan at most $(2+\epsilon)  \beta \cdot \opt^t_{int}$ at all times, with total recourse  $O((\nicefrac{1}{\epsilon^5}) \log (\nicefrac{1}{\epsilon})\log n \log (\nicefrac{r}{\epsilon})) \cdot \optr^{\beta}(\sigma)$. 
\item A randomized dynamic algorithm that maintains an assignment $\psi^t$ with expected makespan  $O(\frac{\log \log n}{\log \log \log n}) \cdot \opt^t_{int}$ for all $t$, with expected total recourse  $O( \log (r)) \cdot \optr^{\beta}(\sigma) + \frac{T}{\poly(n)}$.
\end{enumerate}
Here, $\optr^{\beta}(\sigma)$ is the optimal total recourse for any {\em offline} algorithm which maintains a feasible {\em fractional} solution to all the five constraints~\eqref{eq:kt:makespan:1}-\eqref{eq:kt:makespan:4} throughout the update-sequence $\sigma$, and $r := \max_{j \in \mathcal{J}} |M(j)|$ is  the maximum number of machines that are willing to process any given job.
\end{corollary}

\subsection{Dynamic Minimum Spanning Tree}

\label{sec:mst}

In the minimum spanning tree (MST) problem, we are given a connected graph $G = (V, E)$ with $|V| = n$ nodes, and a cost $c_e \geq 0$ associated with each edge.  Our goal is to compute a spanning tree $\mathcal{T} = (V, E_\mathcal{T})$ of $G$ of minimum total cost $\sum_{e \in E_\mathcal{T}} c_e$. 

In the dynamic setting, at each time-step $t\in [T]$ an edge gets inserted into/deleted from the set $E$, subject to the condition that the input graph $G = (V, E)$ remains connected. Let $E^t$ denote the state of $E$ at time $t$, and let $\opt^t$ denote the value of the minimum cost spanning tree of $G^t = (V, E^t)$. We wish to maintain a spanning tree $\mathcal{T}^t = (V, E_\mathcal{T}^t)$ in $G^t$ whose cost is always within a factor $\beta \geq 1$ of $\opt^t$, such that the total recourse, $\sum_{t \in [T]} \left| E_\mathcal{T}^t \oplus E_\mathcal{T}^{t-1}\right|$ is as small as possible. The fractional version of the dynamic MST problem fits naturally into our framework with the following
\begin{align}
\label{eq:kt:mst}
    K_t & := \left\{x^t\in \mathbb{R}^{\binom{|V|}{2}}_+ \ \middle| \  
    \begin{array}{cc}
         \sum_{e \in \partial^t(S)} x^t_{e}  \geq 1 & \forall \ \emptyset \subset S \subset V, \\
          x_e^t = 0 & \forall e \in \binom{V}{2} \setminus E^t, \\
          \sum_{e \in E^t}  c_e \cdot x_e^t  \leq  \beta \cdot \textsc{Opt}^t. & 
    \end{array} \ \right\}.
\end{align}
Note that the vector $x^t$ has an entry for every {\em potential edge}  $e \in \binom{V}{2}$. If a potential edge $e$ is currently not present in the graph, then we set $x_e^t = 0$. The notation $\partial^t(S)$ denotes the set of all edges in $G^t$ crossing the cut $S$.  We say that an $x^t \in \mathbb{R}^{\binom{V}{2}}_+$ is a {\em fractional MST} in $G^t$, with cost at most $\beta \cdot \opt^t$, iff $x^t \in K_t$ as per~\eqref{eq:kt:mst} above. The main result in this section is summarized in the theorem below. We defer the proof of \cref{th:kt:mst} to \cref{app:th:kt:mst}.

\begin{restatable}[MST Rounding]{theorem}{mstround}
\label{th:kt:mst}
Consider any fractional solution $x^t$ to the dynamic MST problem, and any $\delta > 0, \alpha \geq 1$. Then there exists a randomized dynamic rounding scheme which maintains a spanning tree $\mathcal{T}^t = (V, E_\mathcal{T}^t)$ of $G^t$ such that:
\begin{eqnarray}
\label{eq:th:kt:mst:1}
\expect*{\sum_{t \in [T]} \left\vert E_\mathcal{T}^t \oplus E_\mathcal{T}^{t-1}\right\vert} & \leq &  O\left( \frac{\alpha \log n}{\delta^2}\right) \cdot \sum_{t \in [T]} \norm{x^t - x^{t-1}}_1 +  \frac{T}{O(n^{\alpha})}, \\
\label{eq:th:kt:mst:2}
\sum_{e \in E_\mathcal{T}^t} c_e & \leq & (2+\delta)   \cdot \sum_{e \in E^t} c_e \cdot x_e^t \qquad  \forall t \in [T].
\end{eqnarray} 
\end{restatable}

In this problem, the row-sparsity of the covering constraints is at most $\binom{n}{2} = \Theta(n^2)$. Thus, by \cref{thm:main}, we can maintain a fractional MST $x^t$ in $G^t$ that has cost  $\sum_{e \in E^t} c_e \cdot x_e^t \leq (1+\epsilon) \beta \cdot \opt^t$, and has a $O((\nicefrac{1}{\epsilon}) \log (\nicefrac{n}{\epsilon}))$-competitive total recourse with respect to any offline algorithm which maintains a fractional MST of cost at most  $\beta \cdot \opt^t$ at all times. Setting $\delta = \epsilon$ and $\alpha$ to be some large constant, and combining this with \cref{th:kt:mst}, we get:

\begin{corollary}
\label{cor:kt:mst} For any $\beta \geq 1$, $\epsilon \in (0, 1]$, and any update-sequence $\sigma$ with $T$ updates, there exists a randomized dynamic algorithm that maintains a spanning tree of cost at most $(2+\epsilon)\beta \cdot \opt^t$ at all times, with expected total recourse $O((\nicefrac{1}{\epsilon^3}) \cdot \log n \cdot \log (\nicefrac{n}{\epsilon})) \cdot \optr^{\beta}(\sigma) + \frac{T}{\poly(n)}$. Here, $\optr^{\beta}(\sigma)$ is the optimal total recourse for any {\em offline} algorithm which maintains a fractional MST of cost at most $\beta \cdot \opt^t$ throughout the update-sequence $\sigma$.
\end{corollary}

 We remark that we can easily extend our analysis to allow for other forms of updates, such as  changing the cost of an edge that is being inserted, or vertex insertions/deletions.

\appendix

\section{Reducing Positive Bodies to Positive Halfspaces}\label{sec:reduction_halfspaces}

In this section we show that it is enough to assume that at each time $t$, each positive body $K_t$ is defined by a single covering or packing constraint.  

Suppose we have an algorithm $\mathcal{A}$ for chasing \emph{halfspaces} that satisfies covering constraints fully and satisfies packing constraints up to $\eps$: the algorithm produces $x^t$ such that for all $t \in C$ we have $\inner{c^t}{x} \geq 1$ and for all $t \in P$ we have $\inner{p^t}{x} \leq 1+\eps$.

The following is an algorithm $\mathcal{A}'$ for the that satisfies packing constraints up to violation $\delta$. Given a body $K_t\neq \emptyset$, while there a constraint of $K_t$ violated by at least $\delta / 10$, feed that constraint to algorithm $\mathcal{A}$ run with parameter $\eps = \delta / 20$. With every such fix, algorithm $\mathcal{A}$ moves at least $\Omega(\nicefrac{\delta}{w})$, where $w$ is the maximal non-zero coefficient in the constraint. Since there exists a finite-movement optimal solution that satisfies all the constraints, and the competitive ratio of algorithm $\mathcal{A}$ that chases half-spaces is also finite, we get that the number of such rounds is finite.

Hence, after finitely many rounds, all the constraints in $K_t$ are satisfied up to $\delta/10$. By scaling up the solution by $(1-\delta/10)^{-1}$ (in all times $t$), we may ensure that all covering constraints are fully satisfied, and packing constraints are violated by up to $\delta$. Furthermore, if $\mathcal{A}$ is $c$-competitive for positive halfspace chasing, then the solution output by $\mathcal{A}$ is at most $c$ that of the optimal algorithm that chases the bodies $K_t$. Hence algorithm $\mathcal{A}'$ is also $O(c)$ competitive for positive body chasing.

\section{Proof of \cref{setcover_rounding}}
\label{appendix_a}

\setcovround*

\begin{proof} We prove each part separately.

\textbf{Proof of (1):}  Initialize  $\mathcal{S}^0 = \emptyset$. Subsequently, while handling the update at time $t$, initialize $\mathcal{S}^t = \mathcal{S}^{t-1}$. Next, whenever $x^t_i \geq 1/f$ for some $S_i \notin \mathcal{S}^{t}$, add $S_i$ to $\mathcal{S}^t$. In contrast, whenever the LP-value $x^t_i \leq \nicefrac{1}{2f}$ for some  $S_i \in \mathcal{S}^t$, remove $S_i$ from $\mathcal{S}^t$.

Thus, we always have $\{S_i : x_i \geq \nicefrac{1}{2f} \} \supseteq \mathcal{S} \supseteq \{ S_i : x_i \geq \nicefrac{1}{f}\}$. Since $x$ is a valid fractional set cover, this ensures that $\mathcal{S}$ is a valid integral set cover with cost at most $2f \cdot \sum_{i=1}^m c(S_i) \cdot x^t_i$. 

To bound the recourse, note that we insert/delete a set $S_i$  in $\mathcal{S}$ (thereby incurring a recourse of one) only after its LP-value has changed by at least $\nicefrac{1}{2f}$. In other words, the recourse of our rounding scheme is at most $2f$ times the recourse incurred by the underlying dynamic algorithm which maintains the fractional solution $x$.

\textbf{Proof of (2):}
We maintain two collections of sets $\mathcal{R}^t$ and $\mathcal{B}^t$, and output as our solution $\mathcal{S}^t := \mathcal{R}^t \cup \mathcal{B}^t$. (The reader may think of $\mathcal{R}$ as being the {\em randomly sampled} sets, and $\mathcal{B}$ as the {\em backup} sets.) For each $i \in [m]$, the set $S_i$ has an exponential random variable $\chi_i$ with parameter $\lambda :=   \alpha \log n$ that is sampled once and fixed at the outset of the dynamic process. Throughout the sequence of updates, any every time-step we define $\mathcal{R}^t := \{ S_i \in \mathcal{S} : x^t_i \geq \chi_i \}$. For each element $u \in U$, let $S(u)$ denote a minimum cost set in $\{S_1, \ldots, S_m\}$ which contains $u$. Define $\mathcal{B}^t = \{S(u) \ | \ u \text{ uncovered by } \mathcal{R}^t\}$. 

Clearly $\mathcal{R}^t \cup \mathcal{B}^t$ is a feasible set cover; it remains to analyze the approximation ratio and the recourse. We do so separately for $\mathcal{R}^t$ and $\mathcal{B}^t$.

To analyze $\mathcal{R}^t$, note that for all $i \in [m]$, we have
\begin{align*}
    \Pr[S_i \in \mathcal{R}^t] &= \Pr[x_i^t \geq \chi_i] = 1 - \exp\left( -x_i^t \cdot   \alpha \log n \right) \leq x_i^t \cdot   \alpha \log n.
\end{align*}
Summing over $i\in [m]$, we bound the expected cost of $\mathcal{R}^t$ as $\expect{c(\mathcal{R}^t)} \leq \sum_{i \in [m]} c(S_i) \cdot x_i^t \cdot  \alpha \log n$.
   
For the recourse of $\mathcal{R}$, observe that due to the $t^{th}$ update, the set $S_i$ contributes a recourse of one iff $\chi_i$ lies in an interval of length $\left| x^{t}_{i} -  x^{t-1}_i \right|$. As the exponential distribution is memoryless, this event occurs with probability at most
$1 - \exp\left(- \left|x_i^t - x_i^{t-1} \right| \cdot    \alpha \log n\right) \leq \left|x_i^t - x_i^{t-1} \right| \cdot   \alpha \log n$.  
Summing over all $i \in [m]$, the expected  recourse of $\mathcal{R}$ due to the $t^{th}$ update is given by:
\begin{align*}
\mathbb{E}\left[ \left| \mathcal{R}^t \oplus \mathcal{R}^{t-1}  \right| \right] \leq  \alpha \log n \cdot \norm{x^t - x^{t-1}}_1.
\end{align*}

Moving on to $\mathcal{B}^t$, consider any element $u \in U^t$.  We have: 
\begin{align*}\Pr\left[u \text{ not covered by } \mathcal{R}^t \right]  &=  \prod_{i \in [m] : \, u \in S_i} \Pr[S_i \notin \mathcal{R}^t] = \prod_{i \in [m] : \, u \in S_i} e^{- x_i^t \cdot   \alpha \log n}  \\ & =  \exp\left\{-  \alpha \log n \cdot  \left(\sum_{i \in [m] : \, u \in S_i} x_i^t\right)\right\} \leq \frac{1}{n^\alpha},
\end{align*}
where the last inequality holds since $x_i^t$ is a feasible fractional set cover. Since $|U^t| \leq n$, it follows that the expected number of uncovered (with respect to $\mathcal{R}^t)$ elements in $U^t$ is at most $n \cdot \nicefrac{1}{n^\alpha} \leq \nicefrac{1}{n^{\alpha - 1}}$. Each of these uncovered elements contributes one set to $\mathcal{B}^t$. Accordingly, the expected recourse of $\mathcal{B}$ for the $t^{th}$  update is at most  $\mathbb{E}\left[ \mathcal{B}^t \oplus \mathcal{B}^{t-1} \right] \leq \mathbb{E}\left[ |\mathcal{B}^t| + |\mathcal{B}^{t-1}| \right] = O(\nicefrac{1}{n^{\alpha-1}})$.

To bound the expected cost of $\mathcal{B}^t$, note that each set $B \in \mathcal{B}^t$ is added because of an element $u \in U^t$ that is uncovered by $\mathcal{R}^t$. Since $B = S(u)$ is a cheapest set covering $u$ and  $x^t_i$ is a feasible fractional set cover for $U^t$, we have $c(B) \leq \sum_{i \in [m]} c(S_i) \cdot x_i^t$. By the argument above, we now derive that:
\[\mathbb{E}\left[ c(\mathcal{}B^t) \right] \leq \mathbb{E} \left[ |\mathcal{B}^t| \cdot \sum_{i \in [m]} c(S_i) \cdot x^t_i\right] \leq \frac{1}{n^{\alpha-1}} \cdot \sum_{i \in [m]} c(S_i) \cdot x^t_i = O(\alpha \log n) \sum_{i \in [m]} c(S_i) \cdot x^t_i.\]

To summarize, the total expected cost of $\mathcal{S}^t = \mathcal{R}^t \cup \mathcal{B}^t$ is at most $O(\alpha \log n) \cdot \sum_{i \in [m]} c(S_i) \cdot x^t_i$, and the expected recourse per update is at most $\alpha \log n \cdot  \|x^t - x^{t-1}\|_1 + O(1/\alpha)$.
\end{proof}
\section{Proof of \cref{th:matching:rounding}}
\label{app:matching}

\maxmatchround*

Our rounding scheme consists of two steps. In the first, we  maintain a subgraph $H^t$ of $G^t$, which we call a {\em stabilizer}, such that  the recourse of $H^t$ is bounded with respect to the $\ell_1$ movement of  $x^t$, and   $H^t$  contains a matching whose size is very close to ${\tt Obj}(x^t) := \sum_{e \in E}  x_e^t$.  In the second step, we  maintain a $(1-\delta)$-approximate maximum matching in $H^t$ whose recourse is small with respect to the number of updates to $H^t$.

In the first step, we show:
\begin{lemma}
\label{lm:matching:sparsifier}
We can maintain a subgraph $H = (V, E_H)$ of $G^t$ such that:
\begin{enumerate}[label={(\arabic*)}]
\item    $\mathbb{E}[\mu(H^t)] \geq (1-\delta) \cdot {\tt Obj}(x^t)$ at all times $t$, where $\mu(H^t)$  is the  maximum matching size in $H^t$.
\item  $\displaystyle \expect*{\sum_t \left\vert E_H^{t-1} \oplus E_H^t\right\vert }\leq O\left(\frac{\alpha \log n}{\delta^2} \right) \cdot \sum_t \norm{x^t - x^{t-1}}_1 +  O\left(\frac{T}{n^{\alpha}}\right)$. 
\end{enumerate}
\end{lemma}

We defer the proof of \cref{lm:matching:sparsifier} for a moment, while we show how to complete the proof of the theorem.

\begin{lemma}
\label{lm:matching:recourse}
We can maintain a $(1-\delta)$-approximate maximum  matching $M^t$ in $H^t$, with  $O(\nicefrac{1}{\delta})$ absolute recourse per update in $H^t$.
\end{lemma}

\begin{proof}
We maintain a matching $M^t$ in $H^t$ that does not admit any augmenting path of length less than $1+\Theta(1/\delta)$.  Such a matching $M^t$ is a $(1-\delta)$-approximate maximum matching in $H^t$. If an update in $H^t$ leads to the creation of an augmenting path of length less than $1+\Theta(1/\delta)$ with respect to $M^t$, then we modify $M^t$ by augmenting along that path. This implies an absolute recourse of $O(1/\delta)$ per update in $H^t$. 
\end{proof}

\cref{th:matching:rounding} now follows by combining \cref{lm:matching:recourse} and \cref{lm:matching:sparsifier}.

\begin{proof}[Proof of \cref{lm:matching:sparsifier}]
 We construct the stabilizer $H^t$ as follows. Fix the parameter
\begin{equation}
\label{eq:kappa}
 \kappa := \lceil (100 (\alpha+4) \log n)/\delta^2 \rceil.
\end{equation}
At preprocessing, for each  
$e = (u, v) \in \binom{V}{2}$ and each $i \in [\kappa]$, we draw a value $\chi_{e, i} \in [0, 1]$ uniformly and independently at random. Define the collection of multi-edges \[F^t := \left\{ e_i : e \in \binom{V}{2}, \ i \in [\kappa],  \text{ and } x_e > \chi_{e, i} \right\},\](note that $F^t$ is always a subset of the edge set of $G^t$)
and let $z$ be the weighted multigraph that assigns each edge in $F^t$ the value $((1+\delta)\kappa)^{-1}$. Define $\Gamma^t$ to be the event that $z^t$ represents a valid fractional matching in $F^t$. Finally, set the edge set of $H^t$ to be $F^t$ if $\Gamma^t$ holds, otherwise  set $H^t$ to be an arbitrary maximum matching in $G^t$.

Before proving \cref{lm:matching:sparsifier}, we first show that $z^t$ is a high-value fractional matching with good probability. For convenience, let 
\[{\tt Obj}(z^t) := \sum_{e_i \in F^t}  z^t_{e_i}.\]
Also, for each  edge $e \in E^t$ and $i \in [\kappa]$, let $X_{e, i} \in \{0, 1\}$ be the indicator for the event that $e_i \in E_F^t$. Note that
\begin{equation}
\label{eq:indicator}
\expect{X_{e, i}} = \prob{x_e > \chi_{e,i}} = x_e^t.
\end{equation}

\begin{claim}
\label{cl:matching:validfracmatch} For all time-steps $t$, event $\Gamma^t$ holds with probability at least $1-1/n^{\alpha+3}$.
\end{claim}

\begin{proof}
    Fix any node $v \in V$. Let $\text{deg}^t(v, F)$ denote the number of multi-edges incident on $v$ in $F$.  
Since $z_{e_i} = (1+\delta)^{-1} \cdot (1/\kappa)$ for all $e_i \in E_F^t$, it suffices to show that $\text{deg}^t(v, F) \leq (1+\delta) \kappa$ with sufficiently high probability. Observe that \eqref{eq:indicator} implies
\begin{equation}
\label{eq:indicator:2}
\expect{ \text{deg}^t(v, F)}  = \sum_{e \in E^t(v), \, i \in [\kappa]} \expect{ X_{e, i}} = \sum_{e \in E^{t}(v)} \kappa \cdot x^t_e \leq \kappa.
\end{equation}
Recall the value of $\kappa$ from~\eqref{eq:kappa}.  Applying a Chernoff bound,
\begin{equation}
\label{eq:indicator:3}
\prob{\text{deg}^t(v, F) \leq (1+\delta) \cdot \kappa} \geq 1 - 1/n^{\alpha+4}.
\end{equation}
The claim follows by a union bound over all $n$ nodes in $G^t$.
\end{proof}

\begin{claim}
\label{cl:matching:condval}
    For all time-steps $t$, we have 
    $\expect{{\tt Obj}(z^t) \, | \, \Gamma^t} \geq (1-\Theta(\delta)) \cdot {\tt Obj}(x^t) - 1/n^{\alpha+1}.$
\end{claim}

\begin{proof}
    From~\eqref{eq:indicator} and the fact that $z^t_{e_i} = ((1+\delta) \kappa)^{-1}$,
\begin{align}
\expect{{\tt Obj}(z^t)}
&= \sum_{e \in E^t} \expect*{\sum_{i \in [\kappa] : \, e_i \in E_F^t}  z^t_{e_i}} = \sum_{e \in E^t} (1+\delta)^{-1}  \cdot x_e^t. \geq (1-\Theta(\delta)) \cdot {\tt Obj}(x^t) 
\label{eq:new:121}
\end{align}
Next, since $|E_F^t| \leq \kappa n^2$,  it follows that with probability 1
\begin{equation}
{\tt Obj}(z^t) = |E_F^t| \cdot (1+\delta)^{-1}(1/\kappa) \leq n^2.
\label{eq:new:151}
\end{equation}
By the law of total probability
\begin{align*}
&\expect{{\tt Obj}(z^t) \, | \, \Gamma^t} \geq 
\expect{{\tt Obj}(z^t) \, | \, \Gamma^t} \cdot \prob{\Gamma^t} = \expect{{\tt Obj}(z^t)} - \expect*{{\tt Obj}(z^t) \, | \, \overline{\Gamma^t}} \cdot \prob{\overline{\Gamma^t}}. 
\intertext{Using \eqref{eq:new:121}, \eqref{eq:new:151}, and \cref{cl:matching:validfracmatch}, we conclude}
&\expect{{\tt Obj}(z^t) \, | \, \Gamma^t} \geq (1-\Theta(\delta)) \cdot {\tt Obj}(x^t) - \frac{n^2}{n^{\alpha + 3}} = (1-\Theta(\delta)) \cdot {\tt Obj}(x^t) - 1/n^{\alpha+1}. \qedhere
\end{align*}
\end{proof}

We are now ready to prove part (1) of \cref{lm:matching:sparsifier}. The claim is trivially true if $G^t$ is the empty graph. Otherwise, by the law of total probability,
\begin{align}
\expect{\mu(H^t)} &= \expect{\mu(F^t) | \Gamma^t} \cdot \prob{\Gamma^t} +  \expect{\mu(G^t) | \overline{\Gamma^t}} \cdot \prob{\overline{\Gamma^t}} \nonumber \\
&\geq \expect{{\tt Obj}(z^t) | \Gamma^t} \cdot \prob{\Gamma^t} +  {\texttt{Obj}}(x^t) \cdot \prob{\overline{\Gamma^t}} \label{line:intgap} \\
&\geq (1-\Theta(\delta)) \cdot {\tt Obj}(x^t) - 1/n^{\alpha+1} \label{line:usingcondval} \\
&\geq (1-\Theta(\delta)) \cdot {\tt Obj}(x^t). \label{line:gnonempty}
\end{align}
Step \eqref{line:intgap} follows from the fact that the bipartite matching polytope has no integrality gap, step \eqref{line:usingcondval} from \cref{cl:matching:condval}, and step \eqref{line:gnonempty} since $G^t$ is nonempty and hence $1/n^{\alpha+1} \leq \delta \cdot \texttt{Obj}(x^t)$.

We conclude with the proof of part (2). The recourse of $H^t$ is at most the recourse of $F^t$, plus $n^2$ for every time-step that $\Gamma^t$ does not hold (because in this case we need to replace the graph with a maximum matching in $G^t$ and back which requires at most $2 \cdot {n \choose 2} \leq n^2$ edge insertions). Since $\Gamma^t$ holds with probability $1 - 1/n^{\alpha + 3}$, the expected contribution of switching to $G^t$ and back is at most $1/n^{\alpha}$ per time-step.

To bound the recourse of $F^t$, note that each copy of an edge $e_i \in E_F^{t} \oplus E_F^{t-1}$ iff the random variable $\chi_{e, i}$ (which is uniform in $[0,1]$) lies in an interval of length $\left| x_e^t - x_e^{t-1}\right|$, and hence
\[\expect{\left| E_F^{t} \oplus E_F^{t-1} \right|} = \kappa \cdot \sum_{e \in \binom{V}{2}} \left| x^t_e - x^{t-1}_e\right| = \kappa \cdot \norm{x^t - x^{t-1}}_1.\] 
Summing over all time-steps $t$ and putting everything together,  the total expected recourse of $H^t$ is at most
\begin{align*}
 &\expect*{\sum_t \left\vert E^t_H \oplus E^{t-1}_H \right\vert} \leq \kappa \cdot \sum_{t} \norm{x^t - x^{t-1}}_1 + \sum_t 1/n^{\alpha}. \qedhere
\end{align*}
\end{proof}

\section{Proof of \cref{th:kt:mst}}
\label{app:th:kt:mst}

\mstround*

We follow the same paradigm as in \cref{app:matching} of maintaining a stabilizer, which in this case is a graph that contains a low-cost spanning tree whose recourse is bounded with respect to the $\ell_1$ movement of $x$, and then giving an absolute recourse algorithm with respect to the updates to the stabilizer.

In the first step, we show:
\begin{lemma}
\label{lm:mst:sparsifier}
We can maintain a subgraph $H = (V, E_H)$ of $G^t$ such that:
\begin{enumerate}[label={(\arabic*)}]
\item $\expect*{\mu(H^t)} \leq (2+\delta) \sum_{e \in E^t} c_e \cdot x^t_e$, where $\mu(H^t)$ is the expected cost of the minimum spanning tree of $H^t$.
\item  $\displaystyle \expect*{\sum_t \left\vert E_H^{t-1} \oplus E_H^t\right\vert }\leq O\left(\frac{\alpha \log n}{\delta^2} \right) \cdot \sum_t \norm{x^t - x^{t-1}}_1 +  O\left(\frac{T}{n^{\alpha}}\right)$. 
\end{enumerate}
\end{lemma}

Once again,  we show to complete the proof before giving the stabilizer construction.
\begin{lemma}
\label{lm:mst:recourse} We can maintain a minimum spanning tree $\mathcal{T}$ of $F$ with $O(1)$ absolute recourse per edge insertion/deletion in $F$.
\end{lemma}

\begin{proof}
When edge $e$ gets deleted from $F$, if $e \notin E_\mathcal{T}$,  we don't make any changes to $\mathcal{T}$, otherwise replace $e$ in $E_\mathcal{T}$, with the cheapest edge crossing its fundamental cut. When an edge $e$ is added to $F$, add $e$ to $E_\mathcal{T}$ then delete from $E_\mathcal{T}$ the cheapest edge along the fundamental cycle associated with $e$. In both cases $\mathcal{T}$ remains the MST after the update, and incurs absolute recourse $O(1)$.
\end{proof}

Combining \cref{lm:mst:sparsifier} and \cref{lm:mst:recourse} yields \cref{th:kt:mst}. We turn to showing how to build the stabilizer.

\begin{proof}[Proof of \cref{lm:mst:sparsifier}]
    
Let $\gamma > 0$ be a sufficiently large constant. For every edge ${e \in \binom{V}{2}}$, define
\begin{equation}
\label{eq:mst:kappa}
 p^t_e := \min\left\{1, \frac{100 \gamma (\alpha+3) \log n\cdot x^t_e}{\delta^2} \right\}
\end{equation}
At preprocessing, for each edge $e \in \binom{V}{2}$ we draw a  value $\chi_e \in [0, 1]$ uniformly and independently at random, and define the edge set $F^t := \left\{ e \in \binom{V}{2} : p_e^t > \chi_e \right\}$ (note that $F^t$ is always a subset of the live edges $E^t$). Let $\Gamma^t$ be the event that $F^t$ contains a spanning tree of cost at most $(2+\delta) \cdot \sum_{e \in E^t} c_e \cdot x_e^t$. Finally, define the edge set of $H^t$ to be $F^t$ if $\Gamma^t$ holds, otherwise set $H^t$ to be an arbitrary minimum spanning tree of $G^t$.

Part (1) of \cref{lm:mst:sparsifier} is immediate by construction (since the integrality gap of the cut LP formulation of minimum spanning is at most 2), and it remains to prove part (2). To this end, we show that $H^t$ contains a cheap spanning tree with good probability.
\begin{claim}
\label{lm:mst:key}
For all time-steps $t$, event $\Gamma^t$ holds with probability at least $1-1/n^{\alpha+3}$.
\end{claim}
\begin{proof}
    We follow the proof of Lemma~3.8 in~\cite{CK18arxiv}. Let $(G^t, x^t)$ be the weighted graph with edge weights defined by $x^t$. Define the \emph{strength} of an edge $e$, denoted by $\kappa^t_e$, to be the maximum
    value of $k$ such that a maximal $k$-connected vertex-induced subgraph of $G^t$ contains $e$.

    The mincut in $(G^t, x^t)$ has value $\geq 1$, since $x^t$ is feasible to \eqref{eq:kt:mst}, and therefore $\kappa^t_e \geq 1$ for all  $e \in E^t$. For each edge $e \in E^t$, we now define
\begin{equation}
\label{eq:mst:parameter}
q_e^t := \min\left\{1, \frac{\gamma (\alpha+3) \log n \cdot c_e \cdot x_e^t}{\delta^2 \cdot \sum_{e' \in E^t} c_{e'} \cdot x^t_{e'}} \right\} \text{ and } r_e^t := \max\{p_e^t, q_e^t\}.
\end{equation}

Now, consider the weighted graph $(Z^t, z^t)$ with edge set $E_Z^t := \left\{ e \in \binom{V}{2} : r_e^t > \chi_e \right\}$, where the weight of edge $e \in E_Z^t$ is $z^t_e := (1+\delta) \cdot x^t_e / r^t_e$. Let $\mathcal{E}^t$ be the event that the weighted graph $(Z^t, z^t)$ has mincut at least $1$ and has $\sum_{e \in E_Z^t} c_e \cdot z_e^t \leq (1+\delta) \cdot \sum_{e \in E^t} c_e \cdot x_e^t$. Since the cut LP formulation of minimum spanning tree has integrality gap at most $2$, event $\mathcal{E}^t$ in turn means that $Z^t$ contains an integral spanning tree $\mathcal{T}'$ of cost at most $2 (1+\delta) \cdot \sum_{e \in E^t} c_e \cdot x_e^t$.

Following the analysis in the proof of Lemma 3.8 in~\cite{CK18arxiv}, we derive that $\mathcal{E}^t$ holds with probability at least $1-1/n^{\alpha+3}$. Conditioned on $\mathcal{E}^t$ holding, consider any edge $e^* \in E_Z^t \setminus E_F^t$. We have $p^t_{e^*} \leq \chi_{e^*} <  r^t_{e^*}$. Since $\chi_{e^*} \in [0,1]$, recalling the values of $p_{e^*}^t$ and $r_{e^*}^t$ from~(\ref{eq:mst:kappa}) and~(\ref{eq:mst:parameter}), we infer that
$$\frac{100 \gamma (\alpha+3) \log n \cdot x^t_{e^*}}{\delta^2} \leq \frac{\gamma (\alpha+3)\log n \cdot c_{e^*} \cdot x_{e^*}^t}{\delta^2 \cdot \sum_{e \in E^t} c_{e} \cdot x^t_{e}}.$$
Rearranging the terms in the above inequality, we get $c_{e^*} \geq 100 \cdot \sum_{e \in E^t} c_{e} \cdot x_{e}^t$, which means that $e^*$ is not in $\mathcal{T}'$, the spanning tree of cost $2 (1+\delta) \cdot \sum_{e \in E^t} c_e \cdot x_e^t$ in $Z^t$. In other words, the edges of $\mathcal{T}$ are also contained in $F^t$, which was the claim to be proven. 
\end{proof} 

We conclude with the proof of part (2) of \cref{lm:mst:sparsifier}. The recourse of $H^t$ is at most the recourse of $F^t$, plus $n^2$ for every time-step that $\Gamma^t$ does not hold (because in this case we need to replace the graph with a spanning tree in $G^t$ and back which requires at most $2 \cdot {n \choose 2} \leq n^2$ edge insertions). Since $\Gamma^t$ holds with probability $1 - 1/n^{\alpha + 3}$, the expected contribution of switching to $G^t$ and back is at most $1/n^{\alpha}$ per time-step.

To bound the recourse of $F^t$, note that each edge $e \in E_F^{t} \oplus E_F^{t-1}$ iff the random variable $\chi_{e}$ (which is uniform in $[0,1]$) lies in an interval of length $\left| p_e^t - p_e^{t-1}\right| = O(\alpha \log n / \delta^2) \left| x_e^t - x_e^{t-1}\right|$, and hence
\[\expect{\left| E_F^{t} \oplus E_F^{t-1} \right|} = O\left(\frac{\alpha \log n}{\delta^2}\right) \cdot \sum_{e \in \binom{V}{2}} \left| x^t_e - x^{t-1}_e\right| = O\left(\frac{\alpha \log n}{\delta^2}\right) \cdot \norm{x^t - x^{t-1}}_1.\] 
Summing over all time-steps $t$ and putting everything together,  the total expected recourse of $H^t$ is at most
\begin{align*}
 &\expect*{\sum_t \left\vert E^t_H \oplus E^{t-1}_H \right\vert} \leq O\left(\frac{\alpha \log n}{\delta^2}\right) \cdot \sum_{t} \norm{x^t - x^{t-1}}_1 + \sum_t O(1/n^{\alpha}). \qedhere
\end{align*}
\end{proof}

\section{Lower Bound}
\label{sec:lb}

In this section we show that resource augmentation is essential in order to get a sub-polynomial competitive ratio for positive body chasing. 

\lbmainthm*

The crux of our bound is the following lemma.

\begin{lemma}
\label{lemma:lb_eps}
    No algorithm for positive body chasing can achieve competitive ratio better than \[\max\left\{\frac{\sqrt{n} (1 - 10 \eps \cdot \sqrt{n \log n})}{10},0\right\}.\]
\end{lemma}

Let us first see how this implies \cref{thm:lb_mainthm}.

\begin{proof}[Proof of \cref{thm:lb_mainthm}]
    First note that the $\Omega(\log n)$ lower bound for fractional online set cover (see e.g. \cite{BN05}) implies an $\Omega(\log n)$ bound for positive body chasing, even with $\eps$ resource augmentation when $\eps \in (0,1]$.

    It remains to argue the remaining part of the bound. If $\eps \leq 1 / (20 \sqrt{n \log n})$, then \cref{lemma:lb_eps} implies a lower bound of $\Omega(\sqrt n)$. If on the other hand $\eps > 1 / (20 \sqrt{n \log n})$, then the same construction from \cref{lemma:lb_eps} on a subspace of dimension $n' = 1 / (900 \cdot \eps^2 \log (1/\eps))\leq n$ implies a lower bound of $\Omega(\sqrt{n'}) = \Omega(1/\eps \sqrt{\log (1/\eps)})$.
\end{proof}

Finally, we prove the main lemma.

\begin{proof}[Proof of \cref{lemma:lb_eps}]
Assume that $n = 2^k$ for sufficiently large $k$. Also, assume $\eps \leq 1/10\sqrt{n\log n}$, otherwise there is nothing to prove.
Let $h^i$ be the $i^{th}$ row of the Hadamard matrix of order $n$. Define $S = \sqrt{n} (2\sqrt{\log n} + 1)$, as well as $(x^0,s^0) = (2\sqrt{(\log n)/ n}\cdot \vec 1, \sqrt{n})$. Note that $S = \|(x^0, s^0)\|_1$.

We execute the following instance for $M$ many phases. In each phase the adversary chooses a uniformly random vector $b\in\{-1,1\}^{n}$. A phase is composed of $n$ time-steps, where the body at time $t \in [n]$ is defined as 
\begin{align*}
    K_t &= \left\{(x,s)\in R^{n+1}_+ \ \middle| \ \inner{h^t}{x} + s+\sum_{i=1}^{n}x_i \geq b^t+ S + \inner{h^t}{x^0} \right\} && \text{if } b^t=1 \\
    K_t &= \left\{(x,s)\in R^{n+1}_+ \ \middle| \ \inner{h^t}{x} + s+\sum_{i=1}^{n}x_i \leq b^t+ S + \inner{h^t}{x^0} \right\} && \text{if } b^t=-1.
\end{align*} 
We claim that this is an instance of positive body chasing.
Let us first check that the constraint of each $K_t$ is a covering/packing constraint. 
For the left-hand side of each constraint, the coefficients of each variable are at least zero because entries of  $h^t$ are in $\{+1 , -1\}$ and so the coefficient of $x_i$ is $h^t_i +1 \ge 0$. Also, the coefficient of $s$ is $1$. 
For the right-hand side, we have that $b^t+ S + \inner{h^t}{x^0}\geq 0$ because, by the definition of the Hadamard matrix, $\inner{h^t}{x^0}$ equals to $0$ when $t>0$ and equals to $\|x^0\|$ when $t = 0$. 
It remains to show that $K_t\neq \emptyset$. 
Observe that $(x^0, s^0+b^t)\in K_t$ and, in fact, the constraint is tight: $\inner{h^t}{x_0} + (s^0+b^t)+\sum_{i=1}^{n}x_i = b^t+ S + \inner{h^t}{x^0}$ because $S = \sum_i x_i + s_0$.

\paragraph{Movement cost of $\alg$:}
We argue that any online algorithm must pay at least $\nicefrac{n}{2} \cdot (1- 10 \eps \cdot \sqrt{n\log n})$ per phase in expectation.

For convenience, define $L^t = \inner{h^t}{x^t} + s + \sum_i x_i$ to be the left-hand side of the constraints above, and let $R^t = S + \inner{h^t}{x^0}$. The algorithm needs to satisfy covering constraints fully, and packing constraints up to $(1+\eps)$; if $b^t =1$, the algorithm must choose $x^t$ such that
\begin{align*}
    L^t &\geq R^t + 1,
    \intertext{and if $b^t =-1$, the algorithm must choose $x^t$ such that}
    L^t &\leq (R^t - 1)(1+\eps).
\end{align*}
The term $L^t$ must change by at least half the distance between $R^t+1$ and $(R^t-1)(1+\eps)$ in expectation, i.e.
\[\expect*{\vert L^t - L^{t-1} \vert} \geq  \frac{(R^t + 1) - (R^t-1)(1+\eps)}{2} = 1 - \frac{\eps(R^t-1)}{2}.\]
Since $L^t = \inner{l^t}{(x^t, s^t)}$ where $l^t$ is a vector with coefficients in $[0,2]$, this implies that
\[\expect*{\|(x^t,s^t) - (x^{t-1},s^{t-1})\|_1} \geq \frac{1}{2} - \frac{\eps(R^t-1)}{4} = \frac{1 - 10 \eps \cdot \sqrt{n \log n}}{2}.\]
Therefore, the expected cost of any algorithm during a single phase of $n$ time-steps is at least $\nicefrac{n}{2} \cdot (1- 10 \eps \cdot \sqrt{n\log n})$.

\paragraph{Movement of \opt:} 
We argue that $\opt$ pays at most $5\sqrt{n}$ in expectation per phase. We begin by assuming that $\opt$ always returns to $(x^0,s^0)$ at the end of a phase, since this can only increase the cost of $\opt$.

First, note that since $(x^0, s^0+b^t)\in K_t$, then the optimal solution may move at each time $t$ to this point, and return at the end to $(x^0,s^0)$. The total movement cost of this solution is bounded by $4\sum_{i=1}^{n}|b^t|=4n$. However, we claim that a better solution is available to \opt with high probability. 

Let $\mathcal{E}$ be the event that $\|H^{-1} b\|_\infty\leq 2\sqrt{(\log n)/n}$. Because
\[\|H^{-1} b\|_2 \leq \|H^{-1}\|\cdot \|b\|_2 \leq \frac{1}{\sqrt{n}} \cdot \sqrt{n} = 1,\]
we have that $\|H^{-1} b\|_1 \leq \sqrt{n}\cdot \|H^{-1} b\|_2 \leq \sqrt{n}$.
Define
\begin{align*}
    x^* &= x^0+ H^{-1} b= 2\sqrt{\frac{\log n}{n}}\cdot \vec 1 + H^{-1} b,\\
    s^* &= S-\sum_i x^*_i\geq \sqrt{n} -\|H^{-1} b\|_1
\end{align*}
Now for all $t \in [n]$, 
since $\inner{h^t}{x^*} =  \inner{h^t}{x^0} + b^t$ by definition of $x^*$, we have by design
\[\inner{h^t}{x^*} + s^*+\sum_{i=1}^{n}x^*_i = b^t+ S + \inner{h^t}{x^0}.\]
Additionally, if $\mathcal{E}$ holds, then using the fact that $\|H^{-1} b\|_1 \leq \sqrt{n}$, the point $(x^*, s^*) \in R^{n+1}_+$ and thus $(x^*, s^*) \in \bigcap_{t=1}^n K_t$. Consequently, $\opt$ can move directly from $(x^0,s^0)$ to $(x^*, s^*)$ and back at the end of the phase, paying at most $4\sqrt{n}$. 

It remains to bound the probability of $\mathcal{E}$. The matrix $H^{-1}$ has entries in $\{-\nicefrac{1}{n}, \nicefrac{1}{n}\}$, so each entry of $H^{-1} b$ is a random walk on the line with step size $\nicefrac{1}{n}$. 

\begin{fact}[Chernoff Bound]
    If $Z = \sum_{i=1}^n X_i$ where $X_i$ are independent random variables such that $\prob{X_i = 1} = \prob{X_i = -1} = \nicefrac{1}{2}$, then
    \[\prob{|Z| \geq \sqrt{a n \log n}} \leq 2 n^{-a/2}.\]
\end{fact}
By this fact together with a union bound, event $\|H^{-1} b\|_\infty\leq 2\sqrt{(\log n)/n}$ holds with probability $1 - \nicefrac{2}{n}$. Hence, the expected cost of $\opt$ in total over the phase is $(1-\nicefrac{2}{n}) \cdot 4\sqrt{n} + \nicefrac{2}{n} \cdot 4n \leq 5 \sqrt{n}$, for sufficiently large $n$.

To conclude, summing over $M$ phases and accounting for $\opt$ moving to $(x^0, s^0)$ before the first phase, the competitive ratio of any algorithm is at least 
\[\frac{\expect*{\alg}}{\expect*{\opt}} \geq \frac{M\cdot \nicefrac{n}{2}\cdot (1-10\eps \cdot \sqrt{n \log n})}{5 \cdot M \cdot \sqrt{n} + \|(x^0, s^0)\|_1} \geq \frac{\sqrt{n} (1 - 10 \eps \cdot  \sqrt{n \log n})}{10} \cdot (1- o(1)). \qedhere\]

\end{proof}

\section{Handling Box Constraints}
\label{sec:box}

In this section we show that we can handle box constraints, a special case of packing constraints, without violation. Box constraints are ``static'' packing constraints of the form $x_j\leq b_j$.

In general, we allow to add to the polytope at time $t$ the box constraints $x\leq b$, where $b\in R_{+}^{n}\cup \{\infty\}$. Thus, the more general polytope at time $t$ is:

\[K_{t}^{1+\eps}=\left\{ x^{t}\in R_{+}^{n}\ |\ x\leq b, \ C^{t}x^{t}\geq1, 
\ P^{t}x^{t}\leq1+\eps\right\}\]

We claim that we can simply eliminate these new constraints by adding additional covering constraints.
For each covering constraint $\sum_{j}c_{ij}x_j \geq 1$, and every subset of variables $S$, we add the constraint:
\[\sum_{j\not \in S}c_{ij}x_j \geq \max\left\{0, 1-\sum_{j\in S}c_{ij}\cdot b_j\right\}\]

It is easy to verify that if $x\in R^n_{+}$ satisfies these new constraints, then the vector $x'$, where $x'_j= \min\{x_j, b_j\}$ satisfies the original constraints along with the box constraints. On the other hand, any $x''$ that satisfies the box constraints satisfies all the new covering constraints.

Thus, our algorithm maintains a feasible point to $K_t^{1+\eps}$ but with box constraints removed, and the new covering constraints added instead. The algorithm maintains at time $t$ the vector $x'$, where $x'_j= \min\{x_j, b_j\}$, and its total movement cost is no more than that of $x$. 

Finally, we remark that although the number of such constraints is exponential, it is easy to verify if a solution $x$ satisfies these new constraints by checking that $x$ satisfies the covering constraint with $S=\{j \mid x_j\geq b_j\}$.

{\footnotesize
    \bibliography{dblp,refs}
    \bibliographystyle{alpha}
}
 
\end{document}